\documentclass[superscriptaddress,aps,prx,9pt,twocolumn]{revtex4-2}

\usepackage[T1]{fontenc}
\usepackage{amsmath,amssymb,mathtools,amsthm}
\usepackage{newtxtext}
\usepackage{newtxmath}
\usepackage{braket}
\usepackage{booktabs,colortbl,array,tabularx}
\usepackage[dvipsnames,table]{xcolor}
\usepackage[colorlinks,citecolor=blue,linkcolor=blue,urlcolor=blue]{hyperref}
\usepackage{graphicx}
\usepackage{tikz}
\usetikzlibrary{decorations.pathreplacing,arrows.meta,shapes.geometric,shapes.symbols,shapes.misc,fit}
\usepackage{subfig}
\usepackage{nicematrix}
\newcommand{\safeincludegraphics}[2][]{%
  \IfFileExists{#2}{\includegraphics[#1]{#2}}{%
    \fbox{\parbox[c][3.2cm][c]{0.9\linewidth}{\centering
    Figure file not found:\\[2pt]\texttt{\detokenize{#2}}}}}}

\usepackage[most]{tcolorbox}
\usepackage{xcolor}

\definecolor{boxframe}{gray}{0.8} 
\definecolor{boxbg}{gray}{0.95}

\newtcolorbox{resultbox}[1][]{
  enhanced,
  breakable,
  colback=boxbg,
  colframe=boxframe,
  boxrule=0.6pt,
  arc=2pt,
  left=8pt, right=8pt, top=6pt, bottom=6pt,
  #1
}
\newcommand{\LBonly}[1]{\textcolor{BrickRed}{\mathbf{#1}^{*}}}
\DeclareMathOperator{\Tr}{Tr}

\newtheorem{theorem}{Theorem}
\newtheorem{lemma}[theorem]{Lemma}

\newtheorem{corollary}{Corollary}
\newtheorem{remark}{Remark}

\begin{document}

\title{Optimal Dynamic Cooling of Multiple Qubits}

\author{Mattia Reda}
\affiliation{Dipartimento di Fisica ``A. Volta,'' Universit\`a di Pavia, via Bassi 6, 27100 Pavia, Italy}
\author{Massimiliano F. Sacchi}
\affiliation{CNR -- Istituto di Fotonica e Nanotecnologie, Piazza Leonardo da Vinci 32, I-20133 Milano, Italy}
\affiliation{Dipartimento di Fisica ``A. Volta,'' Universit\`a di Pavia, via Bassi 6, 27100 Pavia, Italy}
\author{Chiara Macchiavello}
\affiliation{Dipartimento di Fisica ``A. Volta,'' Universit\`a di Pavia, via Bassi 6, 27100 Pavia, Italy}
\affiliation{INFN Sezione di Pavia, Via Agostino Bassi 6, I-27100 Pavia, Italy}
\author{Giacomo Guarnieri}
\affiliation{Dipartimento di Fisica ``A. Volta,'' Universit\`a di Pavia, via Bassi 6, 27100 Pavia, Italy}
\affiliation{INFN Sezione di Pavia, Via Agostino Bassi 6, I-27100 Pavia, Italy}

\begin{abstract}
We solve the closed-system problem of cooling $M$ qubits, selected
from $N$ identical thermal qubits, to the lowest common local
temperature allowed by unitarity. The optimal protocol consists of two
conceptually distinct steps. First, an ergotropic permutation assigns
the largest eigenvalues of the initial state to target sectors of 
lowest Hamming weight, thereby minimizing the total target
energy. Second, a target-only complex-Hadamard
transformation within each fixed-Hamming-weight subspace equalizes the
one-qubit target marginals without changing any target-sector
probability or the total energy. Consequently, imposing a common local
temperature costs neither cooling depth nor additional work: the
constrained optimum coincides with the unconstrained minimum
for every $M<N$ and every initial temperature. The complex-Hadamard
correction may nevertheless be costly at the circuit level. We
therefore derive an exact arithmetic criterion for when the same
optimum can be attained by a temperature-independent
computational-basis permutation alone, and exhaustively classify the
resulting finite-size islands of feasibility for $M+2\leq
N\leq128$. At isolated temperatures, further optimal permutations can
arise through numerical cancellations between different thermal
eigenvalue shells. These alternative realizations may reduce
implementation complexity, but they cannot improve the cooling curve
already attained by the universal protocol. We also derive the exact
cooling curve, prove that at least two ancillary qubits are necessary
and sufficient for nontrivial cooling, and show that joint many-target
cooling can strictly outperform parallel single-target
strategies.
\end{abstract}

\maketitle

\section{Introduction}
\label{sec:introduction}
The ability to prepare qubits in a pure target state is a basic
requirement for any quantum information
process~\cite{DiVincenzo2000,NielsenChuang}. Such task is a well-defined thermodynamic resource problem that recurs throughout modern computation~\cite{Terhal2015,Preskill2018,Linden2010,LevyKosloff2012,Mitchison2019,Aamir2025,Krishnan2024,Solfanelli2022,Buffoni2023,Difranco2024}.

Physically cooling qubits, however, often represents a challenging task both energetically and technologically, as it requires to lower the temperature of the entire platform. 
A complementary fruitful route to direct refrigeration is to cool through purification and entropy manipulation. 

The idea of preparing cold qubits by such method rather than by
direct physical refrigeration dates back to the early days of nuclear-magnetic-resonance (NMR) quantum
computation, where it was introduced to boost signal-to-noise ratios through
increased spin polarization~\cite{Sorensen1989,SchulmanVazirani1999,
Boykin2002}. Schulman and Vazirani~\cite{SchulmanVazirani1999} first showed
that a global unitary operation applied to $N$ identical thermal qubits can
concentrate entropy in a subset of them, thereby cooling the others---a
closed-system protocol that, following
Refs.~\cite{Boykin2002,BassmanOftelie2024}, we refer to as \emph{dynamic
cooling}.
In the high-temperature regime relevant for NMR devices, entropy
conservation (Shannon's source coding bound) limits the achievable cooling of a single
target qubit to a factor $1/\sqrt{N}$~\cite{Sorensen1989,Boykin2002,Schulman2005}. This
unfavourable scaling led to the temporary dismissal of the closed-system
approach and motivated open-system extensions, collectively known as heat-bath
algorithmic cooling, in which a subset of qubits is periodically re-thermalized
through the interaction with an external bath~\cite{Boykin2002,Schulman2005,RaeisiMosca2015,
RodriguezBriones2017,Alhambra2019,Park2016,Lin2024}.
On the one hand, these protocols allowed to bypass Shannon's bound and approach a favourable linear increment of the cooling with the number of resources; on the other hand, they require the system to have frequent interactions with an external cold reservoir, which can potentially induce unwanted errors and decoherence. 

The technological landscape has since changed radically. Contemporary quantum
processors---e.g. superconducting circuits, trapped ions, neutral atoms---operate
at very low initial temperatures and high isolation. Meanwhile, fundamental questions about the energetics and state
purification processes have moved to the core of quantum
thermodynamics~\cite{Goold2016,Binder2018,DeffnerCampbell2019,Taranto2023,
MasanesOppenheim2017,Clivaz2019,Taranto2025,Roadmap2026}. 
In this context, dynamic cooling
has recently undergone a genuine renaissance. Bassman Oftelie, De Pasquale,
and Campisi~\cite{BassmanOftelie2024} solved exactly the problem of cooling
one target qubit using $N-1$ identical others, revealing a crossover from the
$1/\sqrt{N}$ scaling at high temperature to a much more favourable $1/N$ scaling in the
low-temperature regime, where modern quantum computing architectures operate. Xuereb \emph{et
al.}~\cite{Xuereb2025} subsequently characterized the fine-grained structure
of the optimal cooling unitary for a single target qubit interacting with $N$
non-identical machine qubits, deriving Carnot-like bounds and reducibility
criteria for the machine. 

However, single-qubit dynamic cooling prevents the possibility to immediately use the cooled output to run a subsequent quantum computation, since any algorithm requires multiple qubits. What would be desirable in fact is the ability to initialize $M<N$ target qubits through a global unitary protocol at the expense of the $N-M$ remaining ones, in order to have a batch of $M$ qubits ready for use.
Intuitively, even a straightforward sequential application of the single-qubit dynamic cooling to reach $M$ cooled target, simply cannot work: since the total amount of resources is fixed to $N$, each implementation progressively depletes the resources, thus preventing the next qubit to be cooled as much as the previous one, and also preventing the possibility to have $M$ target qubits \textit{all at the same temperature}. Last but not least, after a few iterations of the single-qubit dynamic cooling protocol, the 'disposal qubits' are left progressively hotter, thus going back to the original high-temperature unfavourable scaling.

All this therefore elevates the generalization of cooling $M$ qubits collectively through dynamic cooling to a completely independent and still unexplored problem, whose solution is necessary for modern hardware applications. 

In this work, we fully solve this problem for every $(N,M)$ and for every temperature.
In particular, we explicitly find the best unitary that cool down $M$ target qubits with $N-M$ resources, where the M target all share the same temperature.
This introduces two distinct issues, one
structural and one operational. Structurally, minimizing the total
energy of $M$ target qubits does not ensure that the targets form a
uniform resource: different targets may acquire different local
populations. The relevant optimization therefore combines two
requirements:
\begin{enumerate}
    \item the total target energy must be minimal among all global unitaries;
    \item all target qubits must have the same diagonal marginal and hence the same local temperature.
\end{enumerate}
The first requirement is a passive-rearrangement problem that can be equivalently recast as an ergotropic work extraction~\cite{Allahverdyan2004,Pusz1978,Lenard1978}. The second is a symmetry constraint within the highly degenerate target-energy sectors. Treating these requirements separately is the key to the solution.

Our first main result is the characterization of the optimal universal unitary for every $M<N$. First, a permutation $U_{\rm e}$ assigns the
largest eigenvalues of the initial state to the target-energy sectors
of lowest Hamming weight, as prescribed by a greedy filling scheme. This corresponds to finding the ergotropic unitary for $M$ qubits~\cite{Allahverdyan2004}. A subsequent complex-Hadamard transformation $U_{\rm sym}$, chosen for
concreteness as a discrete Fourier matrix in each fixed-weight
subspace, then coherently averages the target strings of equal Hamming
weight. This symmetrization leaves every target-sector probability
unchanged, commutes with the total Hamiltonian, and makes all
one-qubit target marginals identical and diagonal. Consequently, the
optimal cooling unitary $U_{\star}$ subject to the common-temperature constraint produces a uniform cooled register that exactly coincides
with the output of the unconstrained permutation, thus requiring no additional work.

This statement is stronger than the construction of a convenient
protocol. The ergotropic permutation fixes only the mean target excitation, but not the arrangement of the eigenvalues within each fixed target-energy sector. Our protocol exploits this gauge to reach the minimal total target energy while distributing it equally among the targets.
Hence
no unitary can make every target colder than the common temperature
attained here, even if equality of the marginals is not imposed. In
this precise sense, equal-temperature cooling is not a compromise
relative to the unconstrained problem; it is an attainable
representative of its optimum.

The universal construction also separates thermodynamic performance
from implementation complexity and robustness. The complex-Hadamard
correction works for every pair $(N,M)$ and every temperature, but a
coherent implementation may be circuit-expensive. We therefore
determine exactly when a permutation alone can attain the same optimum
for all initial temperatures. A simple divisibility test identifies
the transparent subclass of shell-by-shell Hamming-symmetric
permutations. This condition is sufficient but not exhaustive. The exact physical condition is instead \emph{shell-wise one-body
balance}: within every thermal eigenvalue shell, each target position
must receive the same total incidence. We derive an arithmetic
criterion that characterizes when such a temperature-independent
permutation exists. This weaker physical requirement admits asymmetric
block assignments. We classify all resulting islands of feasibility in the
range $M+2\leq N\leq128$ and identify a further hierarchy of
temperature-specific permutations produced by numerical cancellations
between different shells. Such alternatives may reduce circuit
complexity on a calibrated platform, but they do not change the
cooling limit and, when temperature specific, sacrifice robustness to
drift.

The classification also reveals an arithmetic structure that has no
single-target analogue. The absence of a permutation-only realization
is not caused by insufficient Hilbert-space capacity: the target-energy sectors always contain the required total number of states. The
obstruction is instead the failure of discrete shell multiplicities to
satisfy the incidence balance imposed by the geometry of the target
strings. This interpretation accounts for the sparse, irregular
islands of feasibility and motivates the odd-$M$ conjecture formulated
below.

Operationally, the single-target solution suggests a simple
alternative: divide the register into $M$ disjoint machines and cool
one target in each. We show that this parallel construction loses the
cooperative advantage of a joint transformation. Any parallel protocol
is itself a restricted global unitary and therefore cannot outperform
the global optimum. More sharply, the exact ancillary-resource
threshold implies a regime in which the distinction is absolute:
throughout $M+2\leq N\leq3M-1$, the global protocol cools whereas
equal parallel single-target machines cannot cool at all. At fixed
$M$, the ratio of the global and parallel low-temperature cooling
factors approaches $M$ as $N$ grows. Simultaneous cooling is therefore
a genuinely cooperative finite-resource phenomenon.

Beyond universal optimality, we derive the exact cooling curve over
the full temperature range, prove the low-temperature law $T'\sim
T/\xi_M^\star$, establish that one ancillary qubit can never cool the
target register, and determine a minimum-work representative within
the class of cooling-optimal permutations. The exact curve, the
permutation criteria, and the global-versus-parallel comparison
distinguish what is fixed by thermodynamics from what depends on
implementation.

\begin{resultbox}[title={Summary of Main Results}, coltitle=black]
Dynamic cooling of multiple qubits consists in applying a global unitary to $N$ identical thermal qubits so that $M$ of them (the targets) are cooled to a \emph{common} temperature. The optimal unitary must satisfy two requirements on the final state: the first is minimal total energy of the $M$ target qubits, the second is identical local states of the targets (equal marginals).

We show that the two requirements are compatible at no extra cost, for every $1<M<N$ and every initial temperature. In particular, we obtain the following:
\begin{enumerate}
    \item Characterization of the optimal unitary operation $U=U_{\rm sym}\cdot U_{\rm e}$, defined as the composition of a permutation $U_{\rm e}$ and the complex-Hadamard symmetrization of the target sectors $U_{\rm sym}$ (Sections~\ref{sec:problem} and~\ref{sec:universal-protocol}).
    \item Low-temperature limit: for $T\to0$,
    \[ T'\sim \frac{T}{\xi_M^\star}, \]
    where $\xi_M^\star$ depends only on $(N,M)$ and quantifies the cooling ratio $T'/T$ (Section~\ref{sec:cooling-curve}).
    \item Joint cooling is strictly better than $M$ parallel single-target protocols for $M+2\le N\le 3M-1$, where the latter cannot cool at all. Its low-temperature advantage approaches a factor $M$ as $N$ grows (Section~\ref{sec:global-parallel}).
    \item Permutation-only implementations: a temperature-independent permutation attains the optimum only for specific $(N,M)$. We give an exact arithmetic criterion and classify all cases for $M+2\le N\le128$ (Sections~\ref{sec:qft-free} and~\ref{sec:local-balance}).
\end{enumerate}
\end{resultbox}

The paper is organized as follows. Section~\ref{sec:problem}
formulates the ergotropic
optimization. Section~\ref{sec:universal-protocol} constructs the
complex-Hadamard symmetrization and proves universal
optimality. Section~\ref{sec:cooling-curve} gives the exact cooling
curve, its low-temperature limit, and the ancillary-resource
threshold. Section~\ref{sec:global-parallel} compares the global protocol with parallel strategies. Sections~\ref{sec:qft-free}--\ref{sec:temperature-dependent}
classify permutation-only implementations, from Hamming-symmetric to
temperature-specific. Section~\ref{sec:work} discusses the work cost. We conclude in
Sec.~\ref{sec:conclusion}; technical proofs, constructions, and pseudocode are collected in the appendices.

\section{Problem formulation}
\label{sec:problem}
\begin{figure}
    \centering
    \begin{tikzpicture}[scale=0.9, every node/.style={scale=1}]
  
  \definecolor{lightblue}{RGB}{173, 216, 230}
  \definecolor{lightred}{RGB}{255, 150, 150}
  \colorlet{midviolet}{violet!80}

  \pgfdeclareverticalshading{hotcold}{100bp}{
    color(0bp)=(lightblue);
    color(25bp)=(lightblue);
    color(50bp)=(midviolet);
    color(75bp)=(red);
    color(100bp)=(red)
  }

  \def\spacing{0.5}
  \def\circleSize{0.4cm}
  \def\shortenVal{0pt} 
  
  \node[circle, draw, fill=midviolet, minimum size=\circleSize, inner sep=0pt, outer sep=0pt] (qin0) at (0, 0) {};
  \node[circle, draw, fill=midviolet, minimum size=\circleSize, inner sep=0pt, outer sep=0pt] (qin1) at (0, -\spacing) {};
  \node at (0, -2.3*\spacing) {$\vdots$};
  \node[circle, draw, fill=midviolet, minimum size=\circleSize, inner sep=0pt, outer sep=0pt] (qin4) at (0, -4*\spacing) {};
  \node[circle, draw, fill=midviolet, minimum size=\circleSize, inner sep=0pt, outer sep=0pt] (qin5) at (0, -5*\spacing) {};

  \draw[-, shorten >=\shortenVal] (qin0) -- (2, 0);
  \draw[-, shorten >=\shortenVal] (qin1) -- (2, -\spacing);
  \draw[-, shorten >=\shortenVal] (qin4) -- (2, -4*\spacing);
  \draw[-, shorten >=\shortenVal] (qin5) -- (2, -5*\spacing);

  \node[circle, draw, fill=lightblue, minimum size=\circleSize, inner sep=0pt] (qout0) at (6, 0) {};
  \node at (6, -0.75*\spacing) {$\vdots$};
  \node[circle, draw, fill=lightblue, minimum size=\circleSize, inner sep=0pt] (qout2) at (6, -2*\spacing) {};
  \node[circle, draw, fill=red, minimum size=\circleSize, inner sep=0pt] (qout3) at (6, -3*\spacing) {};
  \node at (6, -3.75*\spacing) {$\vdots$};
  \node[circle, draw, fill=red, minimum size=\circleSize, inner sep=0pt] (qout5) at (6, -5*\spacing) {};

  \foreach \i in {0,2,3,5} {
    \draw[-, shorten <=\shortenVal] (4, -\i*\spacing) -- (qout\i);
  }

  \draw[decorate,decoration={brace,amplitude=5pt}, thick]
    (6.4, 0.2) -- (6.4, -2*\spacing - 0.2)
    node[midway, xshift=0.8cm, align=center] {$M$\\ \footnotesize target\\ \footnotesize system};

  \draw[decorate,decoration={brace,amplitude=5pt}, thick]
    (6.4, -3*\spacing + 0.2) -- (6.4, -5*\spacing - 0.2)
    node[midway, xshift=0.8cm, align=center] {$N-M$\\ \footnotesize ancillary\\ \footnotesize system};

  \draw[fill=white!20] (2, 0.3) rectangle (4, -5*\spacing - 0.3);
  \node at (3, -2.5*\spacing) {\Large $U$};

  \node at (-0.5, -2.5*\spacing) {$T$};
  \node at (6.3, 0.5) {$T'$};

  \pgfmathsetmacro{\legX}{8}
  \pgfmathsetmacro{\legWidth}{0.5}
  \pgfmathsetmacro{\legTop}{0.3}
  \pgfmathsetmacro{\legBottom}{-5*0.5 - 0.3}

  \shade[shading=hotcold]
    (\legX, \legBottom) rectangle (\legX+\legWidth, \legTop);

  \draw (\legX, \legTop) rectangle (\legX+\legWidth, \legBottom);

  \node[right, align=left] at (\legX+\legWidth, \legTop) {\footnotesize hot};
  \node[right, align=left] at (\legX+\legWidth, \legBottom) {\footnotesize cold};
\end{tikzpicture}
    \caption{Dynamic cooling of $M$ target qubits. A global unitary $U$ acts on $N$ identical qubits initially at temperature $T$, cooling the targets to a common temperature $T'$ while transferring entropy and energy to the remaining $N-M$ ancillary qubits.}
    \label{Mtarget}
\end{figure}
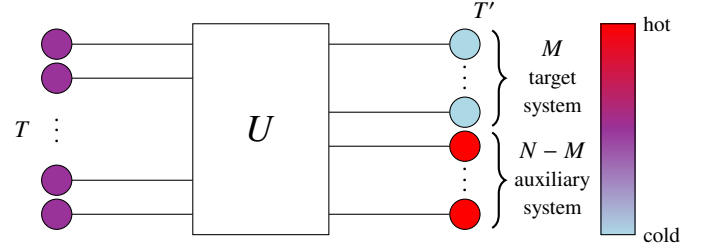

Figure~\ref{Mtarget} summarizes the setting: a global unitary redistributes energy and entropy between the $M$ targets and the remaining $N-M$ qubits. We now formulate the corresponding spectral optimization. Consider $N$ identical, noninteracting qubits with single-qubit Hamiltonian
\begin{equation}
 H_i=\hbar\omega\left(\ket{1}\!\bra{1}_i-\frac{{\mathbb I}_i}{2}\right).
 \label{eq:single-hamiltonian}
\end{equation}
At inverse temperature $\beta$, each qubit is in the Gibbs state
\begin{align}
  & \tau_\beta=P_0\ket{0}\!\bra{0}+P_1\ket{1}\!\bra{1}, \\
   &P_1=\frac{1}{1+e^{\beta\hbar\omega}},
 \qquad P_0=1-P_1,
 \label{eq:single-thermal-state}
\end{align}
and the initial register state is
\begin{equation}
 \rho(P_1)=\tau_\beta^{\otimes N}.
 \label{eq:initial-state}
\end{equation}
We choose the first $M<N$ qubits as targets and denote the remaining $N-M$ qubits as the ancillary register $A$. Here and below, an operator $O_i$ labelled by a qubit index acts non-trivially on qubit $i$ only, i.e.\ it stands for $\mathbb{I}^{\otimes(i-1)}\otimes O_i\otimes\mathbb{I}^{\otimes(N-i)}$. The total Hamiltonian is
\begin{equation}
 H_{\rm tot}=\sum_{i=1}^{N}H_i,
 \label{eq:total-hamiltonian}
\end{equation}
and the target excitation-number operator is
\begin{equation}
 Z_M=\sum_{\alpha=1}^{M}\ket{1}\!\bra{1}_\alpha.
 \label{eq:target-number}
\end{equation}
Since the target Hamiltonian differs from $\hbar\omega Z_M$ only by an additive constant, minimizing the target energy is equivalent to minimizing $\Tr(Z_MU\rho U^\dagger)$.

The eigenvalue associated with an initial computational string containing $k$ excitations is
\begin{equation}
 q_k(P_1)=P_0^{N-k}P_1^k,
 \qquad k=0,\ldots,N,
 \label{eq:qk}
\end{equation}
with multiplicity $\binom Nk$. We call "\textit{shell-k}" the set containing the $\binom Nk$ copies of eigenvalue $q_k$. For $0<P_1<1/2$,
\begin{equation}
 q_0>q_1>\cdots>q_N,
 \label{eq:q-order}
\end{equation}
because $q_{k+1}/q_k=P_1/P_0<1$. The shell ordering is therefore independent of the numerical value of the initial temperature. The eigenvalue $w$ of $Z_M$ has degeneracy
\begin{equation}
 D_w=\binom Mw2^{N-M},
 \qquad w=0,\ldots,M.
 \label{eq:sector-dimension}
\end{equation}

Let $\lambda_1^\downarrow\geq\cdots\geq\lambda_{2^N}^\downarrow$ be the eigenvalues of $\rho$ in nonincreasing order and let $\varepsilon_1^\uparrow\leq\cdots\leq\varepsilon_{2^N}^\uparrow$ be the eigenvalues of $Z_M$ in nondecreasing order, including multiplicities. In eigenbases of $\rho$ and $Z_M$, any unitary gives
\begin{align}
 \Tr(Z_MU\rho U^\dagger)
 &=\sum_{i,j}\varepsilon_i\lambda_j|U_{ij}|^2=\sum_i\varepsilon_i(D\lambda)_i,
 \label{eq:doubly-stochastic-form}
\end{align}
where $D_{ij}=|U_{ij}|^2$ is doubly stochastic. The objective is linear in $D$. By the Birkhoff--von Neumann theorem~\cite{Birkhoff1946,Brualdi2006}, its minimum over the Birkhoff polytope is attained at a permutation matrix, and the rearrangement inequality~\cite{MarshallOlkin} then gives
\begin{equation}
 \min_U\Tr(Z_MU\rho U^\dagger)
 =\sum_{i=1}^{2^N}\varepsilon_i^\uparrow\lambda_i^\downarrow.
 \label{eq:spectral-minimum}
\end{equation}
Thus no coherent unitary can improve on the optimal permutation. 
We refer to this arrangement as $Z_M$-passive since no unitary can lower its mean energy with respect to $Z_M$ alone~\cite{Pusz1978,Lenard1978}, regardless of whether it is passive with respect to the full Hamiltonian.
It is worth stressing that Eq.~\eqref{eq:spectral-minimum} corresponds to the well known solution of the ergotropic work extraction from $M$ qubits~\cite{Allahverdyan2004}; correspondingly, the unitary realizing this optimal permutation will be referred to as the ergotropic unitary $U_{\rm e}$. Since $Z_M$ has degenerate eigenvalues, $U_{\rm e}$ is not unique: several permutations yield the same $Z_M$-passive state. 
Consider a target Hamming-weight sector $w$ and define
\begin{equation}
 \mathcal X_w=\{x\in\{0,1\}^M:w(x)=w\},
 \qquad |\mathcal X_w|=\binom Mw.
 \label{eq:weight-set}
\end{equation}
Finally, let $\Pi_w$ be the projector onto the target-weight-$w$ sector, namely
\begin{equation}
 \Pi_w=\sum_{x\in\mathcal{X}_w}\ket{x}\bra{x},
 \label{eq:projector-weight}
\end{equation}
and, through it, define the sector totals $S_w$ as
\begin{equation}
 S_w(P_1)=\Tr(\Pi_wU_{\rm e}\rho(P_1)U_{\rm e}^\dagger).
 \label{eq:sector-weight}
\end{equation}
The latter are fixed by the ordered spectra and are therefore independent of how degenerate eigenvalues are arranged within a given target sector. 

Equation~\eqref{eq:spectral-minimum} then becomes
\begin{equation}
 \min_U\Tr(Z_MU\rho U^\dagger)
 =\sum_{w=0}^{M}wS_w(P_1).
 \label{eq:sector-spectral-minimum}
\end{equation}
The unconstrained minimum mean target population is consequently given by
\begin{equation}
 p_\star(P_1)=\frac{1}{M}\sum_{w=0}^{M}wS_w(P_1).
 \label{eq:pstar}
\end{equation}
The target-weight sector assignment and the integer multiplicities specifying how many copies of each $q_k$ enter each sector depend only on $(N,M)$, not on $P_1$. This completes the unconstrained mean-energy problem. What is left now is to impose identical target marginals without moving probability between target-energy sectors.
\section{Universal optimal equal-temperature protocol}
\label{sec:universal-protocol}
The passive rearrangement of Sec.~\ref{sec:problem} minimizes the \emph{sum}
of the target excitation probabilities. This requirement alone does not,
however, guarantee that the $M$ target qubits reach the same local
temperature. Indeed, a $Z_M$-passive permutation fixes the total probability
$S_w$ assigned to each target Hamming-weight sector but leaves freedom in how
that probability is distributed among the different target strings within
the sector. 

Let $\rho_{\rm e}=U_{\rm e}\rho U_{\rm e}^\dagger$. The initial state $\rho=\rho(P_1)$, defined in Eq.~\eqref{eq:initial-state}, is diagonal in the computational basis; since the transformation $U_{\rm e}(\cdot)U_{\rm e}^\dagger$ maps diagonal states into diagonal states with respect to the same basis, $\rho_{\rm e}$ can be expressed in block-diagonal form
\begin{equation}
 \rho_{\rm e}=\sum_{x\in\{0,1\}^M}\ket{x}\!\bra{x}\otimes A_x,
 \label{eq:block-state}
\end{equation}
where $x_\alpha\in\{0,1\}$ is the bit of target $\alpha$ and $A_x$ is the corresponding positive operator on the ancillary register. The excited-state population of target $\alpha$ is therefore
\begin{equation}
 p'_{1,\alpha}
 =\sum_{x\in\{0,1\}^M}x_\alpha\Tr A_x.
 \label{eq:local-population-before-symmetrization}
\end{equation}
The sector totals determine only the sum of these populations,
$\sum_\alpha p'_{1,\alpha}=\sum_w wS_w$, and do not generally determine them
individually. Denoting by $\rho'_\alpha$ the reduced final state of target qubit $\alpha$, our equal-temperature requirement is therefore
\begin{equation}
 \rho'_1=\rho'_2=\cdots=\rho'_M,
 \label{eq:equal-local-marginals}
\end{equation}
with each marginal diagonal in the local energy basis. Since the target
qubits have identical gaps, Eq.~\eqref{eq:equal-local-marginals} is equivalent
to requiring a single excited-state population, and hence a single local
temperature, for all targets.

An obstruction appears already in the smallest nontrivial cooling task with
two targets and two ancillary qubits, namely $(N,M)=(4,2)$. The passive
assignment leaves the target-weight-one sector with the spectrum
\begin{equation}
 \left\{q_1\times1,\ q_2\times6,\ q_3\times1\right\}.
 \label{eq:42-middle-spectrum-main}
\end{equation}
This sector consists of the two target strings $01$ and $10$, each associated
with an ancillary block of dimension four. By Eq.~\eqref{eq:local-population-before-symmetrization}, the populations of the target $\alpha=1$ and $\alpha=2$ are equal if and only if $\Tr A_{01}=\Tr A_{10}$ (the permutations that fulfil this condition will be generalized and formalized as \textit{shell-by-shell Hamming-symmetric permutation}, see Sec.~\ref{sec:hamming-symmetric}). In this case, the optimal permutation
would have to divide the multiplicity of every
$q_k$ equally between these two blocks. This is impossible because the
multiplicities of $q_1$ and $q_3$ in Eq.~\eqref{eq:42-middle-spectrum-main}
are odd. Thus, although permutations attain the minimum mean target energy,
none among them leads to target-common-temperature already for the
$4\to2$ task. Appendix~\ref{app:four-to-two-obstruction} gives the complete
block assignment and shows the resulting mismatch explicitly; more general
locally-balanced permutations are characterized later in
Sec.~\ref{sec:local-balance}. 

The present example makes the essential point
that passive energy ordering alone does not enforce a common target
temperature.
The natural solution is to replace the discrete partition of the eigenvalues
inside each degenerate target-energy sector by a coherent average. Starting
from an arbitrary ergotropic permutation, we apply a target-only normalized
complex-Hadamard unitary within every fixed-Hamming-weight subspace. Equal
moduli of its matrix elements ensure that every target string of weight $w$
acquires the same diagonal ancillary block: the arithmetic average of all
blocks in that sector. This operation therefore equalizes the one-qubit
target populations without transferring probability between different
values of $w$. Moreover, it mixes only target strings of equal energy and
acts as the identity on the ancillary register, so it commutes with the total
Hamiltonian. The equal-temperature condition can consequently be imposed
without changing either the passive cooling optimum or the work cost. We now
give the construction explicitly.
With the block decomposition now fixed, consider a target Hamming-weight sector $w$ and recall Eq.\eqref{eq:weight-set}.
Choose a normalized complex-Hadamard unitary $V_w$ on the span of the target strings in $\mathcal X_w$, namely
\begin{equation}
 |\langle y|V_w|x\rangle|^2=\frac{1}{|\mathcal X_w|}
 \qquad (x,y\in\mathcal X_w).
 \label{eq:hadamard-property}
\end{equation}
A convenient choice is the discrete Fourier matrix
\begin{equation}
 (F_{|\mathcal X_w|})_{rs}=\frac{1}{\sqrt{|\mathcal X_w|}}e^{2\pi i rs/|\mathcal X_w|}.
 \label{eq:fourier-matrix}
\end{equation}
Define the target-sector symmetrization
\begin{equation}
 U_{\rm sym}=\bigoplus_{w=0}^{M}
 \left(V_w\otimes\mathbb I_A\right),
 \label{eq:target-symmetrization}
\end{equation}
with $V_0=V_M=1$.

For a fixed sector $w$, the diagonal ancillary block after symmetrization is
\begin{align}
 A'_y
 &=\sum_{x\in\mathcal X_w}
 |\langle y|V_w|x\rangle|^2 A_x \\
 &=\frac{1}{|\mathcal X_w|}\sum_{x\in\mathcal X_w}A_x
 =:\overline A_w,
 \qquad y\in\mathcal X_w.
 \label{eq:block-average}
\end{align}
\begin{figure}[t]
\centering
\begingroup
\setlength{\arraycolsep}{3pt}
\renewcommand{\arraystretch}{1.25}
\resizebox{0.98\columnwidth}{!}{$
\rho_{\rm e}=\left(\begin{array}{c|c|c|c}
\cellcolor{yellow!20}A_{00}&0&0&0\\ \hline
0&\cellcolor{blue!15}A_{01}&0&0\\ \hline
0&0&\cellcolor{red!15}A_{10}&0\\ \hline
0&0&0&\cellcolor{green!15}A_{11}
\end{array}\right)
\xrightarrow{\;U_{\rm sym}\;}
\rho'=\left(\begin{array}{c|c|c|c}
\cellcolor{yellow!20}A_{00}&0&0&0\\ \hline
0&\cellcolor{gray!15}\overline A_1&\cellcolor{orange!15}C&0\\ \hline
0&\cellcolor{violet!15}C^\dagger&\cellcolor{gray!15}\overline A_1&0\\ \hline
0&0&0&\cellcolor{green!15}A_{11}
\end{array}\right),
\quad \overline A_1=\frac{A_{01}+A_{10}}2.
$}
\endgroup
\caption{Complex-Hadamard correction for two targets, here realized by the bidimensional discrete Fourier matrix $F_2$. The unitary $U_{\rm sym}=\mathbb{I}_{A}\oplus (F_2\otimes \mathbb{I}_{A})\oplus\mathbb{I}_{A}$ acts only on the target-weight-one subspace $\operatorname{span}\{\ket{01},\ket{10}\}$. It replaces the two diagonal ancillary blocks by their average and may create the coherence block $C$, while preserving every sector probability and the total energy.}
\label{QFTaction}
\end{figure}
For two targets, Fig.~\ref{QFTaction} illustrates this averaging explicitly in the only nontrivial sector, $w=1$. In general, the construction equalizes the entire diagonal ancillary operators, not merely their traces. Since $S_w=\sum_{x\in\mathcal X_w}\Tr A_x$, one has
\begin{equation}
 \Tr\overline A_w=\frac{S_w}{\binom Mw}.
 \label{eq:average-block-trace}
\end{equation}

\begin{theorem}[Universal optimal equal-temperature cooling]
\label{thm:main}
For every $N>M$ and every $P_1\in[0,1/2]$, the unitary
\begin{equation}
 U_\star=U_{\rm sym}U_{\rm e}
 \label{eq:optimal-unitary}
\end{equation}
produces identical diagonal one-qubit target marginals with common excited population $p_\star(P_1)$ defined in Eq.~\eqref{eq:pstar}. Moreover,
\begin{equation}
 \min_{\substack{U:\,\rho'_1=\cdots=\rho'_M\\
 [\rho'_\alpha,H_\alpha]=0}}
 \bra{1}\rho'_1\ket{1}
 =\min_U\frac{1}{M}\sum_{\alpha=1}^{M}
 \bra{1}\rho'_\alpha\ket{1}
 =p_\star(P_1).
 \label{eq:main-optimum}
\end{equation}
The symmetrization adds no work because
\begin{equation}
 [U_{\rm sym},H_{\rm tot}]=0.
 \label{eq:zero-extra-work}
\end{equation}
\end{theorem}

\begin{proof}
The arbitrary-unitary lower bound has already been established in Eqs.~\eqref{eq:doubly-stochastic-form}--\eqref{eq:sector-spectral-minimum}. For any final state,
\begin{equation}
 \Tr(Z_M\rho')=\sum_{\alpha=1}^{M}
 \bra{1}\rho'_\alpha\ket{1},
 \label{eq:mean-target-energy}
\end{equation}
so the spectral minimum is a lower bound on the sum, and hence on the mean, of the target populations.

It remains to prove that $U_\star$ attains this bound while equalizing all marginals. Fix a target position $\alpha$. Among the strings of weight $w$, exactly $\binom{M-1}{w-1}$ have $x_\alpha=1$. Using Eqs.~\eqref{eq:block-average} and~\eqref{eq:average-block-trace},
\begin{align}
 \bra{1}\rho'_\alpha\ket{1}
 &=\sum_{w=1}^{M}
 \sum_{\substack{x\in\mathcal X_w\\x_\alpha=1}}
 \Tr\overline A_w \\
 &=\sum_{w=1}^{M}\frac{w}{M}S_w
 =p_\star(P_1).
 \label{eq:explicit-marginal-calculation}
\end{align}
The result is independent of $\alpha$, so all target populations are equal. A one-qubit coherence would require an element between two target strings that agree in every position except $\alpha$. Such strings differ in Hamming weight by one, whereas $U_{\rm sym}$ never couples different weights. Every one-qubit target marginal is therefore diagonal.

Finally, $U_{\rm sym}$ mixes only target strings with the same number of target excitations and acts trivially on the ancillary register. It commutes separately with the target and ancillary Hamiltonians, proving Eq.~\eqref{eq:zero-extra-work}. Since the sector weights $S_w$ are unchanged, $U_\star$ saturates the arbitrary-unitary lower bound. This proves Eq.~\eqref{eq:main-optimum}.
\end{proof}

\begin{corollary}[Pareto optimality]
\label{cor:pareto}
No global unitary can bring all $M$ targets to excited-state populations strictly below $p_\star(P_1)$.
\end{corollary}
\begin{proof}
If every target population were strictly smaller than $p_\star(P_1)$, their mean would also be smaller than $p_\star(P_1)$, contradicting Eq.~\eqref{eq:main-optimum}.
\end{proof}

The theorem establishes attainability and optimality at the level of
the target excited-state population. To interpret this population as a
physical cooling temperature, observe that the identity is an
admissible unitary and leaves every target population equal to
$P_1$. Hence,
\begin{equation} p_\star(P_1)\leq
  P_1\leq\frac{1}{2}.
  \label{eq:pstar-positive-temperature}
\end{equation}
Because the final one-qubit target marginals are diagonal, each of
them is therefore a Gibbs state with nonnegative inverse
temperature
\begin{equation} \beta' = \frac{1}{\hbar\omega}
  \log\!\left[ \frac{1-p_\star(P_1)}{p_\star(P_1)} \right] \geq
  0. \label{eq:final-inverse-temperature}
\end{equation}
Strict cooling, $T'<T$, occurs precisely when $p_\star(P_1)<P_1$.

We remark that the operation in Eq.~\eqref{eq:target-symmetrization}
is a Fourier transform on the $\binom Mw$ target strings of fixed
weight, tensored with the ancillary identity. It is not a Fourier
transform on the full $\binom Mw2^{N-M}$-dimensional
target-plus-ancilla sector; the latter could mix ancillary energies
and need not be work-free.
\section{Exact cooling curve and low-temperature limit}
\label{sec:cooling-curve}
Let $c_{w,k}$ denote the number of copies of $q_k$ assigned by the passive rearrangement to target sector $w$. Notably, these integers admit a closed geometric definition. Introduce the cumulative shell multiplicities
\begin{equation}
 B_k=\sum_{j=0}^{k}\binom Nj,
 \qquad B_{-1}=0,
 \label{eq:cumulative-shells}
\end{equation}
and the cumulative target-sector capacities
\begin{equation}
 G_w=2^{N-M}\sum_{r=0}^{w}\binom Mr,
 \qquad G_{-1}=0.
 \label{eq:cumulative-sectors}
\end{equation}

In the ordered list of eigenvalues, shell-$k$ (see below Eq.~\eqref{eq:qk}) occupies the integer
interval $(B_{k-1},B_k]$, whereas target sector $w$ occupies
  $(G_{w-1},G_w]$. Hence, the number of shell-$k$ eigenvalues assigned
    to sector $w$ is the cardinality of their overlap:
\begin{align}
c_{w,k} &= \left| (B_{k-1},B_k]\cap(G_{w-1},G_w] \right|
    \nonumber\\ &= \max\left\{ 0,\,
    \min(B_k,G_w)-\max(B_{k-1},G_{w-1})
    \right\}. \label{eq:explicit-cwk}
\end{align}
This expression is equivalent to the greedy prescription that fills
the target sectors in increasing target weight, always using the
largest remaining thermal eigenvalues. The total probability assigned
to target sector $w$ is therefore
\begin{align} S_w(P_1) &=
  \sum_{k=0}^{N}c_{w,k}q_k(P_1) = \sum_{k=0}^{N}
  c_{w,k}P_0^{N-k}P_1^k.
  \label{eq:sector-polynomial}\\
 p_\star(P_1)&=\frac{1}{M}
 \sum_{w=0}^{M}\sum_{k=0}^{N}
 wc_{w,k}P_0^{N-k}P_1^k.
 \label{eq:exact-cooling-curve}
\end{align}

\subsection{Worked example: $N=6$ and $M=3$}
For $N=6$ and $M=3$, the initial shell multiplicities are $(1,6,15,20,15,6,1)$ and the target-sector capacities are $(8,24,24,8)$. Equation~\eqref{eq:explicit-cwk} gives
\begin{align}
 w=0:&\quad q_0\times1,\ q_1\times6,\ q_2\times1,\nonumber\\
 w=1:&\quad q_2\times14,\ q_3\times10,\nonumber\\
 w=2:&\quad q_3\times10,\ q_4\times14,\nonumber\\
 w=3:&\quad q_4\times1,\ q_5\times6,\ q_6\times1.
 \label{eq:63-sector-assignment}
\end{align}
Accordingly,
\begin{align}
 p_\star(P_1)=\frac{1}{3}\bigl[&14q_2+10q_3
 +2(10q_3+14q_4)\nonumber\\
 &+3(q_4+6q_5+q_6)\bigr].
 \label{eq:63-cooling-curve}
\end{align}
The weight-one and weight-two multiplicities are not divisible shell
by shell by $\binom31=\binom32=3$, so a Hamming-symmetric permutation
does not exist. Nevertheless, the target-only complex-Hadamard
correction makes all three marginals identical without changing
Eq.~\eqref{eq:63-cooling-curve}. This example also lies in the regime
where three parallel two-qubit protocols cannot cool.

\subsection{Low-temperature expansion}
Define
\begin{equation}
 \xi_M^\star=
 \max\left\{\xi\in\{1,\ldots,N\}:
 \sum_{k=0}^{\xi-1}\binom Nk\leq2^{N-M}\right\}.
 \label{eq:xi-definition}
\end{equation}
The target-weight-zero sector has capacity $2^{N-M}$. It contains every shell below $\xi_M^\star$ and
\begin{equation}
 R_{N,M}=2^{N-M}-\sum_{k=0}^{\xi_M^\star-1}\binom Nk
 \label{eq:remainder-zero-sector}
\end{equation}
copies of shell $\xi_M^\star$. Hence the number of shell-$\xi_M^\star$ eigenvalues left for non-zero target weight is
\begin{align}
 L_{N,M}
 &=\binom N{\xi_M^\star}-R_{N,M}\nonumber\
 &=\binom N{\xi_M^\star}-2^{N-M}
 +\sum_{k=0}^{\xi_M^\star-1}\binom Nk.
 \label{eq:leading-multiplicity}
\end{align}
By maximality in Eq.~\eqref{eq:xi-definition},
\begin{equation}
 \sum_{k=0}^{\xi_M^\star}\binom Nk>2^{N-M},
 \label{eq:xi-maximality}
\end{equation}
so $L_{N,M}>0$. The remaining members of this shell first enter the
target weight-one sector. No non-zero-weight sector contains a shell
of lower index. The leading contribution to the common target
population is therefore
\begin{align}
 p_\star(P_1)
 &=C_{N,M}P_1^{\xi_M^\star}
 (1-P_1)^{N-\xi_M^\star}
 +O(P_1^{\xi_M^\star+1})\nonumber\\
 &=C_{N,M}P_1^{\xi_M^\star}
 +O(P_1^{\xi_M^\star+1}),
 \label{eq:low-p-expansion}
\end{align}
with the explicit positive coefficient
\begin{equation}
 C_{N,M}=\frac{L_{N,M}}{M}
 =\frac{1}{M}\left[
 \binom N{\xi_M^\star}-2^{N-M}
 +\sum_{k=0}^{\xi_M^\star-1}\binom Nk
 \right]>0.
 \label{eq:leading-coefficient}
\end{equation}
At low temperature, $P_1=e^{-\beta\hbar\omega}[1+o(1)]$ and $p_\star=e^{-\beta'\hbar\omega}[1+o(1)]$. Thus
\begin{equation}
 \beta'\hbar\omega=
 \xi_M^\star\beta\hbar\omega-\log C_{N,M}+o(1),
 \label{eq:beta-asymptotics}
\end{equation}
and, more precisely,
\begin{equation}
 \frac{T'}{T}=
 \left[
 \xi_M^\star-
 \frac{\log C_{N,M}}{\beta\hbar\omega}
 +o(\beta^{-1})
 \right]^{-1}.
 \label{eq:temperature-finite-correction}
\end{equation}
Consequently,
\begin{equation}
 T'\sim\frac{T}{\xi_M^\star}.
 \label{eq:temperature-asymptotics}
\end{equation}
For $M=1$, this recovers the known single-target low-temperature
behavior~\cite{BassmanOftelie2024}. The asymptotic equality should not be
confused with a finite-temperature identity;
Eq.~\eqref{eq:temperature-finite-correction} displays the leading
logarithmic correction.

\subsection{Cooling threshold}
\begin{theorem}[Exact ancillary-resource threshold]
\label{prop:cooling-threshold}
For $0<P_1<1/2$, nontrivial equal-temperature cooling is possible if and only if
\begin{equation}
 N-M\geq2.
 \label{eq:two-ancilla-threshold}
\end{equation}
Equivalently,
\begin{equation}
 p_\star(P_1)<P_1
 \quad\Longleftrightarrow\quad
 N-M\geq2.
 \label{eq:threshold-population-form}
\end{equation}
\end{theorem}
\begin{proof}
We prove necessity and sufficiency separately.

\emph{Necessity.} Let $N-M=1$. The target-weight-$w$ sector has
dimension
\begin{equation}
 D_w=2\binom{N-1}{w}.
 \label{eq:one-ancilla-sector-dimension-main}
\end{equation}
A $Z_{N-1}$-passive assignment fills this sector with
\begin{equation}
 \binom{N-1}{w}\ \text{copies of }q_w
 \quad\text{and}\quad
 \binom{N-1}{w}\ \text{copies of }q_{w+1}.
 \label{eq:one-ancilla-passive-assignment-main}
\end{equation}
The sector capacity is saturated, and the assignment reproduces every shell multiplicity because
\begin{equation}
 \binom{N-1}{k}+\binom{N-1}{k-1}=\binom Nk.
 \label{eq:one-ancilla-pascal-main}
\end{equation}
It also respects $q_0>q_1>\cdots>q_N$, and is therefore passive. Hence
\begin{align}
 \min_U\Tr(Z_{N-1}U\rho U^\dagger)
 &=\sum_{w=0}^{N-1}w\binom{N-1}{w}(q_w+q_{w+1})\nonumber\\
  &=(N-1)P_1.
 \label{eq:one-ancilla-result-main}
\end{align}
Thus the minimum mean target population equals $P_1$. The identity attains equality, so one ancillary qubit cannot cool.

\emph{Sufficiency.} Assume $N-M\geq2$. A basis state with one target excitation and no ancillary excitations has probability $q_1$, whereas a state with no target excitations and two ancillary excitations has probability $q_2$. Both states exist because $M\geq1$ and $N-M\geq2$. Since $q_1>q_2$, exchanging these populations lowers the target excitation by $q_1-q_2>0$. The initial arrangement is not $Z_M$-passive, and therefore
\begin{equation}
 p_\star(P_1)<P_1.
 \label{eq:strict-cooling-two-ancillas}
\end{equation}
Theorem~\ref{thm:main} then realizes this smaller mean as the common population of all targets.
\end{proof}

\section{Global vs parallel protocols}
\label{sec:global-parallel}
We now compare the cooperative global optimum with a restricted strategy that cools the targets independently. A parallel benchmark must specify how the register is divided. We consider $M$ disjoint groups, one target per group, constrained to produce the same final target population. The natural temperature-independent benchmark uses equal group size
\begin{equation}
 n=\left\lfloor\frac{N}{M}\right\rfloor,
 \label{eq:parallel-group-size}
\end{equation}
and leaves the remaining $N-Mn$ qubits unused. Each group implements the optimal $n$-qubit single-target protocol. Denoting its population by $p_\star(n,1;P_1)$, we set
\begin{equation}
 p_{\rm par}(N,M;P_1)=p_\star(n,1;P_1).
 \label{eq:parallel-definition}
\end{equation}
More generally, any collection of disjoint single-target protocols followed by local operations that deliberately degrade colder outputs to a common population remains a particular admissible global unitary. Therefore the unrestricted global optimum obeys
\begin{equation}
 p_{\rm glob}(N,M;P_1)
 =p_\star(N,M;P_1)
 \leq p_{\rm par}(N,M;P_1)
 \label{eq:global-parallel}
\end{equation}
pointwise.

The inequality is strict throughout
\begin{equation}
 M+2\leq N\leq3M-1.
 \label{eq:exclusive-window}
\end{equation}
Indeed, Theorem~\ref{prop:cooling-threshold} shows that the global protocol cools whenever $N\geq M+2$, whereas every equal parallel group in this window has $n<3$ and a single-target protocol with fewer than three qubits cannot cool. This is a genuinely cooperative regime: only the joint transformation lowers the common target temperature.

Outside the cooperative window, strictness is not automatic. Nevertheless, whenever the equal-partition parallel protocol itself produces cooling, equality with the global optimum occurs only in one case.

\begin{theorem}[Unique equality case for equal-partition parallel cooling] 
\label{thm:unique-parallel-equality} 
Let $M\geq2$, set 
\begin{equation} n=\left\lfloor\frac{N}{M}\right\rfloor, \qquad N=Mn+r, \qquad 0\leq r<M, 
\end{equation} 
and assume that the equal-partition parallel protocol cools, equivalently $n\geq3$. Then, for $0<P_1<1/2$, 
\begin{equation} 
p_{\rm glob}(N,M;P_1)=p_{\rm par}(N,M;P_1) 
\end{equation} 
if and only if 
\begin{equation} 
(N,M)=(6,2). 
\end{equation} 
For this pair, 
\begin{equation} 
p_{\rm glob}(6,2;P_1) = p_{\rm par}(6,2;P_1) = 3P_1^2-2P_1^3. 
\end{equation} 
\end{theorem} 
\begin{proof} 
By Eq.~\eqref{eq:q-order}, populations decrease strictly with the thermal-shell index. 
We show that, except for $(N,M)=(6,2)$, the parallel output is not $Z_M$-passive: it assigns a larger population to a state of 
target weight one than to a state of target weight zero. Exchanging these two populations strictly lowers the mean target excitation, and therefore the parallel protocol cannot attain the global optimum. 

Suppose first that $n=2s+1$. In each group, the optimal single-target protocol assigns shells $0,\ldots,s$ to target state zero and shells $s+1,\ldots,2s+1$ to target state one. The parallel output therefore contains a target-weight-zero state of shell index $Ms+r$ and a target-weight-one state of shell index $s+1$. Absence of a population inversion requires 
\begin{equation} 
Ms+r\leq s+1, \qquad\text{that is,}\qquad (M-1)s+r\leq1. 
\end{equation} 
Since $M\geq2$, $s\geq1$, and $r\geq0$, this is possible only for $M=2$, $s=1$, and $r=0$, which gives $(N,M)=(6,2)$. 

Suppose instead that $n=2s$. The central shell $s$ is shared between the target-zero and target-one blocks of each optimal single-target protocol. Hence the parallel output contains a target-weight-zero state of shell index $Ms+r$ and a target-weight-one state of shell index $s$. Since $Ms+r>s$ for $M\geq2$, these two states form a strict population inversion. Thus equality is impossible for even $n$. It follows that equality can occur only for $(N,M)=(6,2)$. For this pair, the optimum becomes (see Appendix~\ref{app:worked-examples} for details)
\begin{equation} 
p_{\rm glob}(6,2;P_1) = p_{\rm par}(6,2;P_1) = 3P_1^2-2P_1^3, 
\end{equation} 
which proves the claim. 
\end{proof}

The most relevant consequence of Theorem~\ref{thm:unique-parallel-equality} is that for every $(N,M)\neq{(6,2)}$ the global strategy always outperform the parallel one with $(6,2)$ being the only nontrivial equality case for the equal-partition benchmark. 

Figure~\ref{M2plots} illustrates these
possibilities for four representative two-target registers.
\begin{figure*}[t]
    \subfloat{\safeincludegraphics[width=0.45\linewidth]{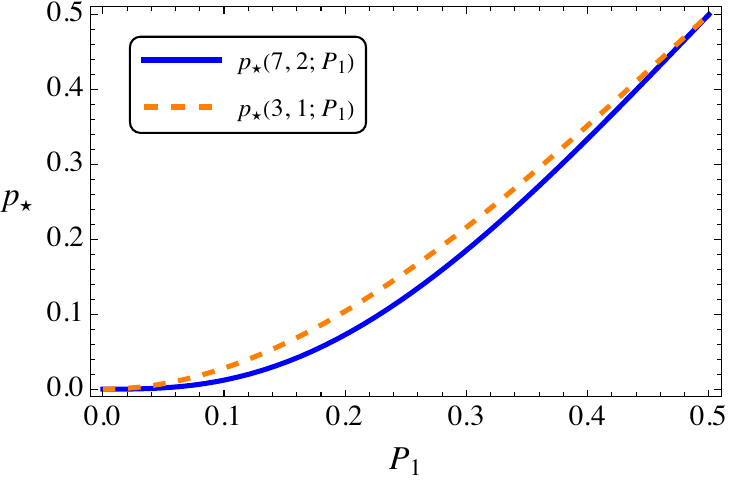}}\hfill
    \subfloat{\safeincludegraphics[width=0.45\linewidth]{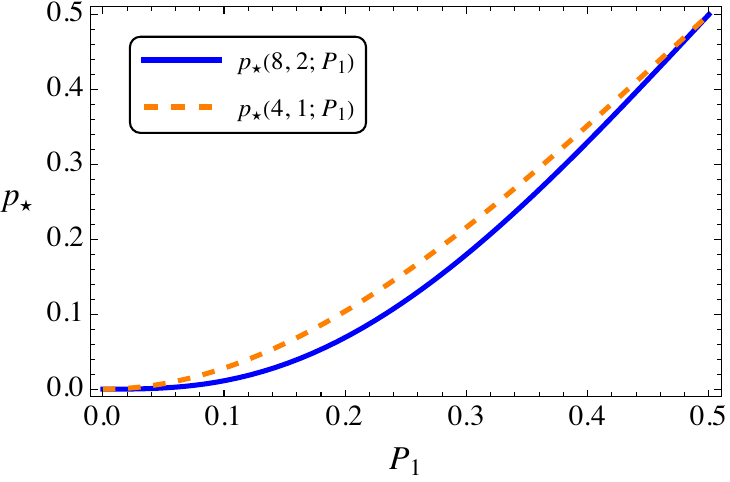}}\\
    \subfloat{\safeincludegraphics[width=0.45\linewidth]{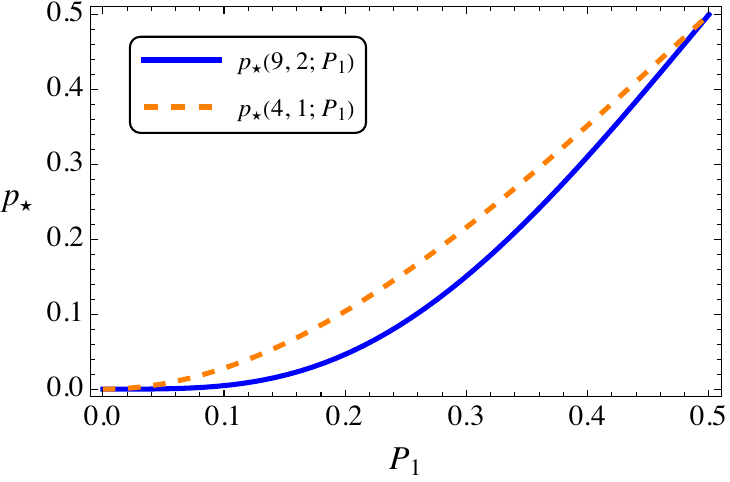}}\hfill
    \subfloat{\safeincludegraphics[width=0.45\linewidth]{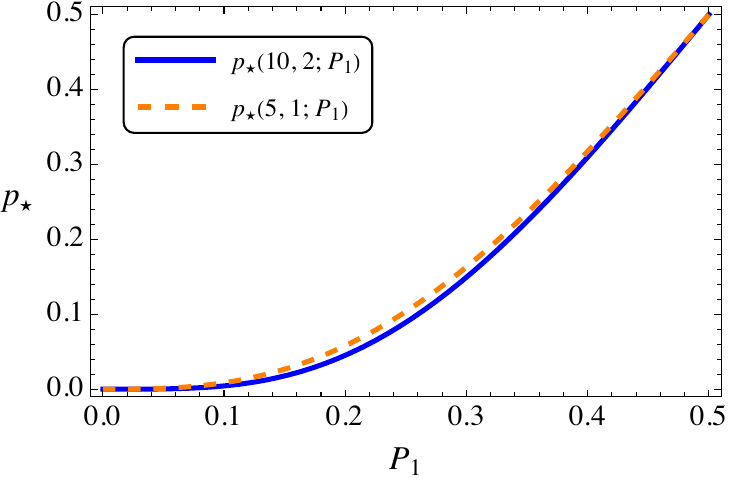}}
    \caption{Optimal common target population for $M=2$. The global protocol (blue) is compared with two equal parallel single-target protocols (orange) for $N=7,8,9,10$. The curves use $p_\star(N,2;P_1)$ and $p_\star(\lfloor N/2\rfloor,1;P_1)$, respectively.}
    \label{M2plots}
\end{figure*}

At low temperature, the global and equal-parallel cooling factors are $\xi_M^\star(N)$ and $\xi_1^\star(\lfloor N/M\rfloor)$, respectively. A Chernoff--Hoeffding estimate for the binomial lower tail gives
\begin{equation}
 \xi_M^\star(N)\geq
 \left\lfloor\frac{N}{2}-
 \sqrt{\frac{NM\ln2}{2}}\right\rfloor.
 \label{eq:hoeffding-xi}
\end{equation}
A derivation, including the integer rounding in Eq.~\eqref{eq:hoeffding-xi}, is given in Appendix~\ref{app:chernoff-hoeffding}.
At fixed $M$, $\xi_M^\star(N)=N/2-O(\sqrt N)$, while $\xi_1^\star(\lfloor N/M\rfloor)=N/(2M)+O(1)$. Hence
\begin{equation}
 \frac{\xi_M^\star(N)}
 {\xi_1^\star(\lfloor N/M\rfloor)}
 \longrightarrow M
 \qquad(N\to\infty),
 \label{eq:parallel-asymptotic-ratio}
\end{equation}
so the global protocol asymptotically reaches a temperature lower by a factor $M$ in the low-temperature regime.
Worked examples for $(N,M)=(5,2)$, $(9,3)$, $(6,3)$, and the equality case $(6,2)$ are given in Appendix~\ref{app:worked-examples}.

\section{When can the complex-Hadamard correction be omitted?}
\label{sec:qft-free}

Having determined the universal cooling limit, we now turn from
performance to implementation. The universal construction is optimal,
but implementing the complex-Hadamard blocks may be circuit-expensive.
We therefore seek optimal permutations that already possess sufficient
symmetry.

\subsection{Hamming-symmetric optimal permutations}
\label{sec:hamming-symmetric}

For $x\in\mathcal X_w$, let $n^{(w)}_{x,k}$ be the number of
shell-$k$ eigenvalues $q_k$ assigned to block $A_x$. We call an
optimal permutation \emph{shell-by-shell Hamming-symmetric} when
\begin{equation}
 n^{(w)}_{x,k}=n^{(w)}_{y,k}
 \quad\text{for all }k\text{ and all }x,y\in\mathcal X_w.
 \label{eq:hamming-symmetry}
\end{equation}
This implies equal block operators up to an ancillary basis
permutation ($A_x$ and $A_y$ are diagonal and contain the same
multiset of eigenvalues, including multiplicities) and,
in particular,
\begin{equation}
 \Tr A_x=\Tr A_y
 \quad\text{whenever }w(x)=w(y),
 \label{eq:equal-traces}
\end{equation}
for every value of $P_1$.

Let $c_{w,k}$ be the passive sector multiplicities in
Eq.~\eqref{eq:sector-polynomial}.  There are $|\mathcal X_w|=\binom Mw$ blocks in sector
$w$.  Equation~\eqref{eq:hamming-symmetry} is feasible if and only if
\begin{equation}
 c_{w,k}\equiv0\pmod{|\mathcal X_w|}
 \qquad\text{for every }w,k.
 \label{eq:divisibility}
\end{equation}
Indeed, every block must receive exactly $c_{w,k}/|\mathcal X_w|$ copies of $q_k$.

\subsection{Divisibility algorithm}
\label{sec:algorithm}

The following greedy procedure determines the integers $c_{w,k}$ and tests
Eq.~\eqref{eq:divisibility}.

\begin{enumerate}
\item Initialize shell reservoirs
\begin{equation}
 r_k=\binom Nk,
 \qquad k=0,\ldots,N,
 \label{eq:algorithm-reservoir}
\end{equation}
and set $k=0$.

\item For $w=0,\ldots,M$, set the remaining sector capacity to
\begin{equation}
 C_w=\binom Mw2^{N-M}.
 \label{eq:algorithm-capacity}
\end{equation}

\item While $C_w>0$, take
\begin{equation}
 t=\min\{r_k,C_w\},
 \label{eq:algorithm-take}
\end{equation}
set $c_{w,k}=t$, update $r_k\leftarrow r_k-t$ and
$C_w\leftarrow C_w-t$, and increase $k$ whenever $r_k=0$.

\item At every allocation, test
\begin{equation}
 t\equiv0\pmod{\binom Mw}.
 \label{eq:algorithm-test}
\end{equation}
If every test succeeds, a Hamming-symmetric optimal permutation exists and
the complex-Hadamard correction is unnecessary.  If one test fails, no
shell-by-shell Hamming-symmetric permutation exists.
\end{enumerate}

The test is temperature independent because the ordering of the shells
is temperature independent.  It classifies the strong
Hamming-symmetric subclass only; failure of
Eq.~\eqref{eq:algorithm-test} does not exclude more general locally
balanced permutations. Figure~\ref{stateforM} visualizes the
target-weight sectors underlying both tests.
\begin{figure}
    \centering
    \begin{tikzpicture}[scale=0.8]

  \def\ha{0.4}
  \def\hb{1.0}
  \def\hc{1.8}
  \def\hd{2.6}

  \definecolor{col0}{RGB}{70,130,180}
  \definecolor{col1}{RGB}{100,180,100}
  \definecolor{col2}{RGB}{200,200,60}
  \definecolor{colMh}{RGB}{220,160,50}
  \definecolor{colM2}{RGB}{230,130,50}
  \definecolor{colM1}{RGB}{210,70,70}
  \definecolor{colM}{RGB}{160,30,30}

  \pgfmathsetmacro{\oa}{0}
  \pgfmathsetmacro{\ob}{\oa+\ha}
  \pgfmathsetmacro{\oc}{\ob+\hb}
  \pgfmathsetmacro{\od}{\oc+\hc+0.7}
  \pgfmathsetmacro{\oe}{\od+\hd+0.7}
  \pgfmathsetmacro{\of}{\oe+\hc}
  \pgfmathsetmacro{\og}{\of+\hb}
  \pgfmathsetmacro{\N}{\og+\ha}

  \draw[thick] (0,0) rectangle (\N,-\N);

  \fill[col0!70] (\oa,-\oa) rectangle (\oa+\ha,-\oa-\ha);
  \draw[thick]   (\oa,-\oa) rectangle (\oa+\ha,-\oa-\ha);

  \fill[col1!70] (\ob,-\ob) rectangle (\ob+\hb,-\ob-\hb);
  \draw[thick]   (\ob,-\ob) rectangle (\ob+\hb,-\ob-\hb);
  \node[font=\small] at (\ob+\hb/2, -\ob-\hb/2)
        {};

  \fill[col2!70] (\oc,-\oc) rectangle (\oc+\hc,-\oc-\hc);
  \draw[thick]   (\oc,-\oc) rectangle (\oc+\hc,-\oc-\hc);
  \node at (\oc+\hc/2, -\oc-\hc/2)
        {$A_{w_2}$};

  \node[font=\Large] at (\oc+\hc+0.35, -\oc-\hc-0.35) {$\ddots$};

  \fill[colMh!70] (\od,-\od) rectangle (\od+\hd,-\od-\hd);
  \draw[thick]    (\od,-\od) rectangle (\od+\hd,-\od-\hd);
  \node at (\od+\hd/2, -\od-\hd/2)
        {$A_{w_{\lfloor M/2 \rfloor}}$};

  \node[font=\Large] at (\od+\hd+0.35, -\od-\hd-0.35) {$\ddots$};

  \fill[colM2!70] (\oe,-\oe) rectangle (\oe+\hc,-\oe-\hc);
  \draw[thick]    (\oe,-\oe) rectangle (\oe+\hc,-\oe-\hc);
  \node at (\oe+\hc/2, -\oe-\hc/2)
        {$A_{w_{M-2}}$};

  \fill[colM1!70] (\of,-\of) rectangle (\of+\hb,-\of-\hb);
  \draw[thick]    (\of,-\of) rectangle (\of+\hb,-\of-\hb);
  \node[font=\small] at (\of+\hb/2, -\of-\hb/2){};

  \fill[colM!70] (\og,-\og) rectangle (\og+\ha,-\og-\ha);
  \draw[thick]   (\og,-\og) rectangle (\og+\ha,-\og-\ha);

\end{tikzpicture}
    \caption{Block structure induced by target Hamming weight. Sector $w$ contains $\binom Mw$ target strings, each labeling an ancillary block of dimension $2^{N-M}$. The passive rearrangement fills these sectors in increasing $w$.}
    \label{stateforM}
\end{figure}

\section{Beyond equal traces: the true local-balance condition}
\label{sec:local-balance}

The divisibility test is deliberately strong because it equalizes the
complete shell content of every block. Equal target temperatures
require less. For each target Hamming-weight sector $w$, define the
sector-resolved incidence at target position $\alpha$ 
\begin{equation}
 d^{(w)}_{\alpha,k}:=\sum_{x\in\mathcal X_w}x_\alpha n^{(w)}_{x,k},
 \label{eq:sector-shell-incidence}
\end{equation}
where $n^{(w)}_{x,k}$ is defined above. The total incidence of
shell $k$ at target position $\alpha$ is
\begin{equation}
 d_{\alpha,k}:=\sum_{w=0}^{M}d^{(w)}_{\alpha,k}
 =\sum_{w=0}^{M}\sum_{x\in\mathcal X_w}x_\alpha n^{(w)}_{x,k}.
 \label{eq:total-shell-incidence}
\end{equation}
For a diagonal permutation output,
\begin{align}
 p'_{1,\alpha}(P_1)&=\sum_xx_\alpha\Tr A_x
 =\sum_{k=0}^{N}d_{\alpha,k}q_k(P_1).
 \label{eq:local-population-shells}
\end{align}
The functions $q_k(P_1)=P_1^k(1-P_1)^{N-k}$ are linearly
independent. Indeed, dividing a putative identity
$\sum_ka_kq_k(P_1)=0$ by $(1-P_1)^N$ and setting $z=P_1/(1-P_1)$ gives
a polynomial $\sum_ka_kz^k$ that vanishes on an interval, so every
coefficient $a_k$ vanishes. Consequently, a single permutation
produces equal target populations for every $P_1$ if and only if
\begin{equation}
 d_{\alpha,k}=d_{\gamma,k}\qquad\text{for all }\alpha,\gamma,k.
 \label{eq:shellwise-local-balance}
\end{equation}
We call Eq.~\eqref{eq:shellwise-local-balance} \emph{shell-wise one-body balance}. A sector-resolved realization, used below, is obtained by making $d^{(w)}_{\alpha,k}$ independent of $\alpha$ for every pair $(w,k)$.

Hamming symmetry implies one-body balance, but the converse is false. The former fixes the full distribution over target strings inside each weight sector; the latter fixes only its one-body marginals. Therefore, the modular table generated by Sec.~\ref{sec:algorithm} answers the operationally useful question ``when can the complex-Hadamard correction be omitted using a fully Hamming-symmetric permutation?'' but not the more general question ``when can it be omitted using any temperature-independent permutation?''

The shell-wise balance equations admit an exact characterization.
\begin{theorem}[Exact temperature-independent local-balance criterion]
\label{thm:exact-local-balance}
Let $c_{w,k}$ be the passive multiplicity of shell $k$ in the target
Hamming-weight-$w$ sector. A temperature-independent optimal
permutation with identical one-qubit target marginals exists if and
only if
\begin{equation}
 M\mid wc_{w,k}\qquad
 \text{for every }w=1,\ldots,M-1\text{ and }k=0,\ldots,N.
 \label{eq:exact-local-balance-criterion}
\end{equation}
Whenever such a permutation exists, it can be chosen sector-resolved
balanced, with
\begin{equation}
 d^{(w)}_{\alpha,k}=d_{w,k}:=\frac{wc_{w,k}}{M}
 \qquad\text{for every }\alpha.
 \label{eq:required-sector-shell-degree}
\end{equation}
Consequently,
\begin{equation}
 d_{\alpha,k}=\sum_{w=0}^{M}d_{w,k}
 =\frac1M\sum_{w=0}^{M}wc_{w,k},
 \label{eq:total-from-sector-degrees}
\end{equation}
which is independent of $\alpha$.
\end{theorem}
Total local balance first implies
\begin{equation}
 M\mid\sum_{w=0}^{M}wc_{w,k}\qquad\text{for every }k.
 \label{eq:aggregate-shell-divisibility}
\end{equation}
The cumulative-congruence lemma in
Appendix~\ref{app:cumulative-congruence}, expressed using the quantities
$B_k$, $G_w$, and $c_{w,k}$ defined in Sec.~\ref{sec:cooling-curve},
shows that these aggregate conditions force
Eq.~\eqref{eq:exact-local-balance-criterion}. Sufficiency follows from
the almost-regular edge-coloring construction in
Appendix~\ref{app:exact-local-balance-proof}, based on Bryant's
theorem~\cite{Bryant2016}. For general background on uniform
hypergraphs, see Ref.~\cite{VanLintWilson2001}.

\begin{remark}[Cross-sector cancellation]
\label{rem:cross-sector-cancellation}
An individual shell-wise locally balanced assignment need not be sector-resolved:
nonzero differences $d^{(w)}_{\alpha,k}-d^{(w)}_{\gamma,k}$ may cancel
after summing over $w$. The theorem is an existence statement: if any
temperature-independent locally balanced optimum exists, a possibly different sector-resolved
balanced optimum can be constructed. Thus cross-sector cancellation
may occur within a solution, but it creates no additional feasible
pair $(N,M)$.
\end{remark}

The two temperature-independent permutation criteria therefore differ
in the relevant divisor:
\begin{equation}
 \begin{array}{lll}
 \text{Hamming symmetry:} && \displaystyle \binom Mw\mid c_{w,k},\\[4pt]
 \text{local-balance existence:} && \displaystyle M\mid wc_{w,k}.
 \end{array}
 \label{eq:criterion-comparison}
\end{equation}

For deciding existence, Theorem~\ref{thm:exact-local-balance}
replaces the integer-feasibility problem by the arithmetic test in
Eq.~\eqref{eq:exact-local-balance-criterion}. When an explicit block
assignment is required, the constructive integer formulation in
Appendix~\ref{app:local-balance-pseudocode} returns the sector-resolved
multiplicities $n^{(w)}_{x,k}$.

Theorem~\ref{thm:exact-local-balance} turns the exact temperature-independent permutation problem into an elementary integer test. Applying it exhaustively over
\begin{equation}
 M+2\leq N\leq128
 \label{eq:classification-range}
\end{equation}
gives the complete classification in
Table~\ref{tableFeasibility}. Every listed pair admits a single
temperature-independent optimal permutation, so the complex-Hadamard
correction is unnecessary. Black entries satisfy the stronger
Hamming-symmetric divisibility condition in
Eq.~\eqref{eq:divisibility}. Asterisked red bold entries fail that
condition but admit a sector-resolved one-body-balanced representative. The latter class
contains $58$ pairs. No locally balanced pair occurs for an omitted
value of $M$.
\begin{table*}[t]
\centering
\small
\renewcommand{\arraystretch}{1.15}
\begin{tabular}{cl}
\toprule
\textbf{Targets $M$} & \textbf{System sizes $N$ admitting a temperature-independent locally balanced optimum}\\
\midrule
$2$ & $5$--$7$, $9$--$14$, $17$--$29$, $33$--$59$, $65$--$121$\\
$3$ & $9$--$14$, $27$--$47$, $63$, $81$--$128$\\
$4$ & $10$, $\LBonly{18\text{--}19}$, $\LBonly{25}$, $34$--$35$, $49$--$51$, $\LBonly{52\text{--}53}$, $66$--$68$, $\LBonly{69}$, $81$--$83$, $97$--$112$\\
$5$ & $25$--$28$, $\LBonly{29\text{--}32}$, $33$--$39$, $75$--$83$, $\LBonly{125\text{--}128}$\\
$6$ & $\LBonly{18\text{--}20}$, $\LBonly{57\text{--}58}$\\
$7$ & $\LBonly{49\text{--}54}$, $55$, $\LBonly{56\text{--}77}$\\
$8$ & $67$, $\LBonly{98\text{--}100}$\\
$11$ & $\LBonly{121\text{--}128}$\\
\bottomrule
\end{tabular}
\caption{Complete classification of the permutation-only islands of
  feasibility in the range $M+2\leq N\leq128$. All listed pairs admit
  a temperature-independent locally balanced optimal permutation and,
  by Theorem~\ref{thm:exact-local-balance}, a sector-resolved balanced
  representative; hence the complex-Hadamard correction is
  unnecessary. Black entries also satisfy the stronger shell-by-shell
  Hamming-symmetric divisibility condition,
  Eq.~\eqref{eq:divisibility}. Red bold entries marked by an asterisk
  fail that condition but satisfy the exact local-balance condition,
  Eq.~\eqref{eq:exact-local-balance-criterion}.  Values of $M$ with no
  admitted system size are omitted.}
\label{tableFeasibility}
\end{table*}
\begin{remark}[Strict separation between the criteria]
The pair $(N,M)=(18,4)$ is the first one that fails the
Hamming-symmetric divisibility test but admits shell-wise one-body
balance. The separation is already visible in the target-weight-two
sector, whose nonzero passive shell multiplicities are
\begin{equation}
 c_{2,8}=24842,
 \qquad
 c_{2,9}=48620,
 \qquad
 c_{2,10}=24842.
 \label{eq:184-main-sector-multiplicities}
\end{equation}
Since this sector contains $\binom42=6$ target strings and all three
multiplicities are congruent to $2$ modulo $6$, the Hamming-symmetric
criterion fails. Nevertheless, the one-body criterion holds because
$4\mid 2c_{2,k}$ for $k=8,9,10$. Thus, the Hamming-symmetric class is
a strict subset of the temperature-independent locally balanced
class. An explicit block assignment is given in
Appendix~\ref{app:strict-separation}.
\end{remark}

\subsection{Numerical threshold and odd-$M$ conjecture}
\label{sec:odd-M-conjecture}
The islands of feasibility suggest an additional arithmetic threshold that is not captured by the exact divisor test alone. We evaluated the criterion in Theorem~\ref{thm:exact-local-balance} using exact integer arithmetic for every $2\leq M\leq20$ and every $M+2\leq N\leq M^2$. The only feasible pairs in this triangular region are
\begin{equation}
 \begin{split}
 &(N,M)=(9,3),(10,4),(25,5),\\
 &(18,6),(19,6),(20,6),(49,7),\\
 &(121,11),(169,13),(289,17),(361,19).
 \end{split}
 \label{eq:numerical-islands-below-M2}
\end{equation}
In particular, no odd $M\leq20$ has an island of feasibility with $N<M^2$. For the odd primes
\begin{equation}
 M=3,5,7,11,13,17,19,
 \label{eq:checked-odd-primes}
\end{equation}
the first feasible value in this range is exactly $N=M^2$. The odd
composite cases $M=9$ and $M=15$ have no feasible value even at
$N=M^2$. By contrast, the exceptions below the quadratic threshold
occur for the even target numbers $M=4$ and $M=6$, namely
$(N,M)=(10,4)$ and $(N,M)=(18,6),(19,6),(20,6)$.

This evidence motivates the following conjecture.

\textbf{Odd-$M$ conjecture.} For every odd $M\geq3$ and every
 $M+2\leq N<M^2$, no temperature-independent locally balanced optimal
 permutation exists. Hence the universal equal-temperature optimum requires
 the complex-Hadamard correction throughout this region.

The conjecture concerns a single permutation that works for every
initial temperature. It does not exclude temperature-specific optimal
permutations at isolated values of $P_1$, where different shells may
cancel numerically as described in
Sec.~\ref{sec:temperature-dependent}. A proof would require showing
that, for odd $M$ and $N<M^2$, at least one passive shell overlap
violates $M\mid wc_{w,k}$.

\section{Temperature-dependent permutation protocols}
\label{sec:temperature-dependent}

The preceding criterion concerns a single permutation that works for every initial temperature. If the permutation may instead be tailored to the initial state, the shell-wise constraints can be relaxed. At a fixed initial population $P_1=p$, equal local populations require only
numerical equality of the sums in Eq.~\eqref{eq:local-population-shells}.
Different shells may then compensate.  Such a cancellation need not persist
when $p$ changes and is not captured by the shell-wise conditions.

For example, consider $(N,M)=(8,2)$, which is not present in
Table~\ref{tableFeasibility}. Each target string labels an ancillary
block of dimension $2^{8-2}=64$, so the target-weight-zero sector has
capacity $64$ and the target-weight-one sector, which contains the two
strings $01$ and $10$, has capacity $2\times64=128$. The passive
ordering fills the weight-zero sector with the largest
eigenvalues. The complete shells $k=0,1,2$ contribute
\begin{equation}
 \binom80+\binom81+\binom82=1+8+28=37
 \label{eq:82-lower-shell-count}
\end{equation}
states. The remaining $64-37=27$ places are filled with $27$ of the
$\binom83=56$ copies of $q_3$. Hence $56-27=29$ copies of $q_3$ remain when
the weight-one sector begins. This sector then receives all
$\binom84=70$ copies of $q_4$ and needs a further
\begin{equation}
 128-29-70=29
 \label{eq:82-upper-boundary-count}
\end{equation}
copies of $q_5$ to reach its capacity. Therefore its spectrum is
\begin{equation}
 \{q_3\times29,q_4\times70,q_5\times29\}.
 \label{eq:82-spectrum}
\end{equation}
The equal boundary multiplicities are also consistent with the symmetry
$\binom8k=\binom8{8-k}$ of the binomial spectrum. The distribution
\begin{align}
 A_{01}&\leftarrow\{q_3\times12,q_4\times50,q_5\times2\},
 \label{eq:82-left}\\
 A_{10}&\leftarrow\{q_3\times17,q_4\times20,q_5\times27\}
 \label{eq:82-right}
\end{align}
has trace difference
\begin{equation}
 \Tr A_{01}-\Tr A_{10}
 =-5q_3+30q_4-25q_5.
 \label{eq:82-difference}
\end{equation}
Writing $z=P_1/(1-P_1)$ gives
\begin{equation}
 \Tr A_{01}-\Tr A_{10}
 =-5q_3(1-z)(1-5z),
 \label{eq:82-factorization}
\end{equation}
which vanishes at the nontrivial value $P_1=1/6$.  The permutation is still
energy-optimal because it redistributes eigenvalues only inside the fixed
weight-one sector.  Away from $P_1=1/6$, the same permutation remains optimal
for the mean target energy but no longer produces equal local populations.

Temperature-dependent permutations can therefore enlarge the set of
permutation-only implementations.  They cannot improve the optimal cooling
curve.  Indeed, Theorem~\ref{thm:main} provides a single
$P_1$-independent unitary attaining the unconstrained spectral minimum for
every $P_1$.  Any temperature-dependent unitary $V_{P_1}$ producing a common
population $\widetilde p(P_1)$ obeys
\begin{equation}
 M\widetilde p(P_1)
 =\Tr(Z_MV_{P_1}\rho V_{P_1}^\dagger)
 \geq Mp_\star(P_1).
 \label{eq:temperature-dependent-bound}
\end{equation}
It may tie the universal optimum at a special temperature, but cannot
surpass it.

\section{Work cost}
\label{sec:work}
The previous sections determine the minimum target population independently of implementation cost. We now evaluate the energetic cost of an optimal transformation. For a cyclic implementation, the supplied average work is
\begin{equation}
 W=\Tr\!\left[H_{\rm tot}
 (U_\star\rho U_\star^\dagger-\rho)\right].
 \label{eq:work}
\end{equation}
Let $n$ be the total excitation-number operator. Since the Hamiltonian differs from $\hbar\omega n$ only by a constant,
\begin{equation}
 W=\hbar\omega\left(
 \langle n\rangle_{\rm f}-NP_1\right).
 \label{eq:work-excitation-form}
\end{equation}
Writing the final excitation number as the sum of target and ancillary contributions gives
\begin{equation}
 \frac{W}{\hbar\omega}
 =M(p_\star-P_1)
 +\left[\langle n_A\rangle_{\rm f}-(N-M)P_1\right].
 \label{eq:work-split}
\end{equation}
The first term is the energy removed from the targets, while the second is the energy gained by the ancillary register. Their difference is supplied by the external control.

Because $[U_{\rm sym},H_{\rm tot}]=0$, the complex-Hadamard correction contributes no additional work. The work is determined entirely by the chosen $Z_M$-passive permutation. Among permutations with the same target-sector assignment, the final ancillary energy is minimized by assigning, within each fixed target weight, the largest remaining probabilities to the smallest ancillary excitation numbers. If $m_{k,w,a}$ denotes the number of shell-$k$ eigenvalues assigned to final target weight $w$ and ancillary weight $a$, then this minimum-work permutation representative has
\begin{equation}
 \frac{W_{\rm perm}^{\min}}{\hbar\omega}
 =\sum_{k,w,a}m_{k,w,a}q_k(w+a)-NP_1.
 \label{eq:minimum-permutation-work}
\end{equation}
Figures~\ref{fig:workHeatmapP02} and~\ref{fig:workHeatmapP04} show the resulting minimum permutation work for two representative initial populations.
\begin{figure}[t]
    \centering
    \safeincludegraphics[width=0.95\linewidth]{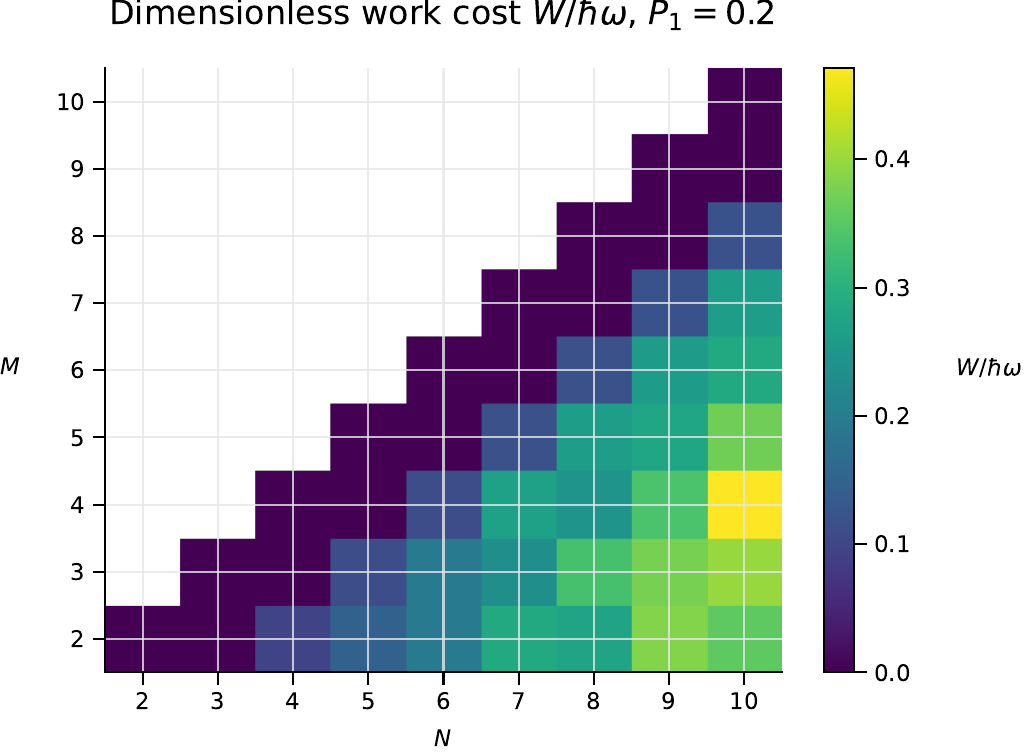}
    \caption{Minimum permutation work $W_{\rm perm}^{\min}/\hbar\omega$ for $2\leq M\leq N\leq10$ at $P_1=0.2$. The complex-Hadamard correction is energy preserving and does not change these values.}
    \label{fig:workHeatmapP02}
\end{figure}
\begin{figure}[t]
    \centering
    \safeincludegraphics[width=0.95\linewidth]{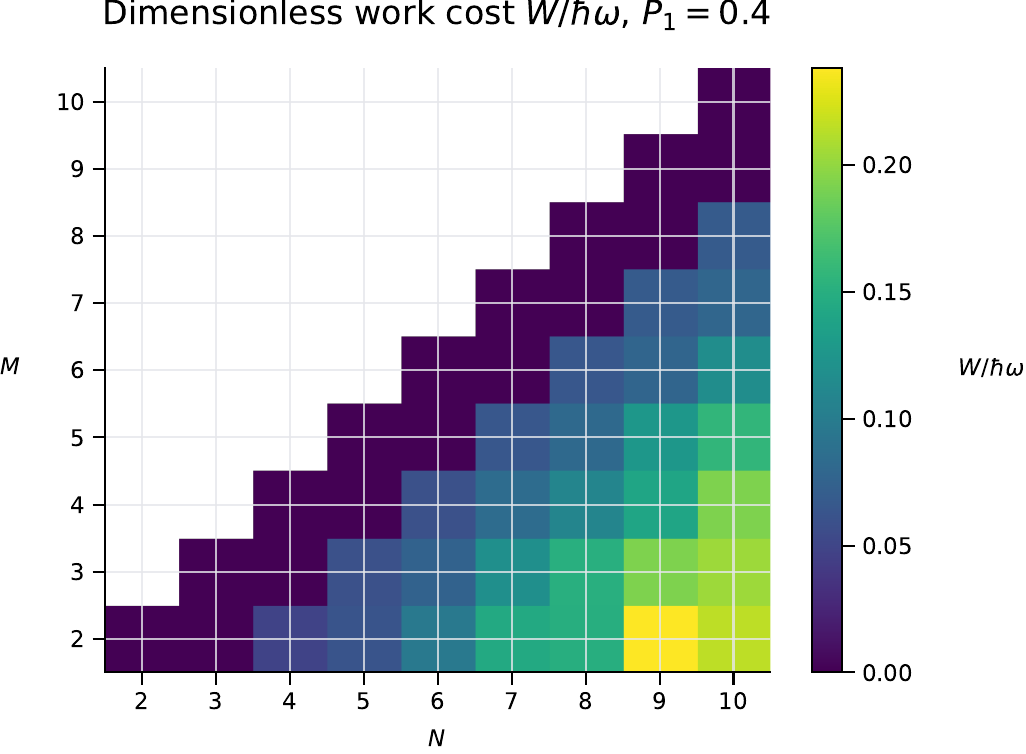}
    \caption{Minimum permutation work $W_{\rm perm}^{\min}/\hbar\omega$ for $2\leq M\leq N\leq10$ at $P_1=0.4$. The smaller spectral bias near $P_1=1/2$ reduces the work required by the passive rearrangement.}
    \label{fig:workHeatmapP04}
\end{figure}

\section{Conclusions and outlook}
\label{sec:conclusion}
We have solved the closed-system dynamic-cooling problem for an
arbitrary number $M$ of target qubits drawn from $N$ identical thermal
qubits. The solution is universal: it applies to every $N>M$ and every
initial temperature. Its structure is simple. A ergotropic permutation
first minimizes the total target energy by assigning the largest
eigenvalues of the initial state to target sectors of lowest Hamming
weight. A target-only complex-Hadamard transformation then symmetrizes
each fixed-weight sector. Because the second step mixes only
equal-energy target strings and leaves the ancillary register
untouched, it preserves the passive target-energy minimum and commutes
with the total Hamiltonian. The target marginals become identical and
remain diagonal. Thus, the common-temperature constraint entails
neither a loss of cooling nor an additional work cost.

This optimality statement is stronger than the existence of a
convenient protocol. The passive permutation constrains the mean
target excitation.
Our construction then reaches the $Z_M$-passive configuration while making all target populations equal. Hence no unitary
can make every target colder than the temperature attained here, even
if the common-temperature requirement is abandoned. Equal-temperature
cooling is therefore not a compromise imposed on the unconstrained
optimum; it is an exactly attainable representative of it.

The exact cooling curve is fixed by the passive sector multiplicities
$c_{w,k}$ and is valid throughout the full temperature range. In the
low-temperature regime it reduces to the universal asymptotic law
$T'\sim T/\xi_M^\star$, with $\xi_M^\star$ determined by the binomial
cumulative condition in Eq.~\eqref{eq:xi-definition}. This formula
recovers the known single-target behavior~\cite{BassmanOftelie2024}
and shows that cooling is impossible when only a single ancillary
qubit is available. The latter statement was established independently
of the common-temperature constraint: with one ancilla, the minimum
mean target population is the initial population itself.

The comparison with independent single-target protocols reveals a
genuinely cooperative advantage. A parallel protocol is only a
restricted member of the set of global unitaries and therefore cannot
outperform the global optimum. In the window $M+2\leq N\leq3M-1$, the
global construction cools while the parallel strategy is completely
inert. Outside this window, the pair $(N,M)=(6,2)$ attains equality with the
equal-parallel benchmark, so strict superiority cannot be claimed
universally at finite $N$. Nevertheless, at fixed $M$ the low-temperature ratio of
global to parallel cooling factors approaches $M$ as $N$ increases.

A central conceptual outcome is the separation between thermodynamic
optimality and implementation. The universal complex-Hadamard
correction guarantees optimal equal-temperature cooling, but different
hardware platforms may favor permutation-only realizations. The
relevant hierarchy is
\begin{equation}
 \begin{split}
& \{\text{Hamming-symmetric permutations}\}\\
 &\subsetneq
 \{\text{temperature-independent locally balanced permutations}\}
 \\
 &\subseteq
 \{\text{temperature-specific optimal permutations}\}
 \\
 &\subseteq
 \{\text{universal optimal unitaries}\}.
 \end{split}
 \label{eq:hierarchy}
\end{equation}
The first class is certified by the divisibility algorithm of
Sec.~\ref{sec:algorithm}. The second is the exact temperature-independent permutation class and
is characterized by the arithmetic criterion in
Eq.~\eqref{eq:exact-local-balance-criterion}. Table~\ref{tableFeasibility} exhaustively
classifies both classes for $M+2\leq N\leq128$ and shows that $58$
pairs are locally balanced without being Hamming-symmetric. Although an
individual temperature-independent locally balanced permutation may
exhibit cross-sector cancellation, Theorem~\ref{thm:exact-local-balance}
guarantees a sector-resolved representative whenever this class is
nonempty. The third class may exploit cancellations between distinct
thermal shells at isolated temperatures. All three classes attain, at best, the same cooling
curve already reached by the universal construction.

Several questions remain open. On the practical side, the
target-sector complex-Hadamard transformations should be compiled into native
gate sets, and their depth, connectivity requirements, and robustness
to noise should be compared with permutation-only and approximate
implementations, extending the circuit-oriented analyses of
Refs.~\cite{BassmanOftelie2024,Difranco2024}. The minimum-work
representative among ergotropic permutations is obtained by additionally
ordering the ancillary excitations, but whether the same cost is
minimal over all coherent representatives of the cooling optimum
remains to be determined. On the structural side, the exact arithmetic
criterion exposes sparse islands of feasibility and suggests the
odd-$M$ conjecture of Sec.~\ref{sec:odd-M-conjecture}: below $N=M^2$,
coherent symmetrization appears unavoidable for odd target
number. Proving or disproving this conjecture would clarify the
asymptotic geometry of the temperature-independent permutation
class. Extensions to nonidentical gaps and arbitrary local dimensions
would connect the present many-target construction with the
single-target machine criteria of Ref.~\cite{Xuereb2025}. Finally, the
coherent symmetrization can generate correlations within degenerate
target sectors; understanding whether these correlations can be
exploited rather than discarded may link dynamic cooling to broader
resource-theoretic questions concerning coherence, correlations, and
finite-resource
thermodynamics~\cite{Baumgratz2014,PerarnauLlobet2015,Taranto2023}.

The main conclusion is therefore both operational and fundamental: for
finite registers of identical thermal qubits, the lowest common
temperature compatible with unitarity is exactly attainable for any
number of targets and at any initial temperature. The discrete
structure of the spectrum determines how the optimum may be
implemented, but it does not obstruct the optimum itself.

\begin{acknowledgments}
G.G. acknowledges support from the Ministero dell'Universit\`a e della
Ricerca under the ``Rita Levi-Montalcini'' grant and from INFN.
The authors acknowledge the use of ChatGPT (OpenAI) and Claude (Anthropic) for developing some
technical steps of the proof of Theorem~\ref{thm:exact-local-balance} and for helping in the construction
of the $(18,4)$ case. All mathematical statements and proofs were independently derived and checked by the authors, who take full responsibility for the content of the paper.

\end{acknowledgments}

\bibliography{references}

\onecolumngrid

\appendix

\section{The permutation obstruction for the $4\to2$ task}
\label{app:four-to-two-obstruction}
Here we give the explicit obstruction referred to at the beginning of
Sec.~\ref{sec:universal-protocol}. For $N=4$, the thermal spectrum is
\begin{equation}
 \left\{q_0\times1,\ q_1\times4,\ q_2\times6,
 q_3\times4,\ q_4\times1\right\},
 \qquad q_k=P_0^{4-k}P_1^k.
 \label{eq:42-full-spectrum}
\end{equation}
With $M=2$, each target string labels an ancillary block of dimension
$2^{N-M}=4$. The target Hamming-weight sectors have dimensions $4$, $8$, and
$4$. The $Z_2$-passive assignment places the four largest eigenvalues in
$A_{00}$ and the four smallest ones in $A_{11}$:
\begin{align}
 A_{00}&\leftarrow\left\{q_0\times1,\ q_1\times3\right\},
 \label{eq:42-A00}\\
 A_{11}&\leftarrow\left\{q_3\times3,\ q_4\times1\right\}.
 \label{eq:42-A11}
\end{align}
The spectrum remaining for the two target-weight-one blocks is
\begin{equation}
 \Gamma_{w=1}
 =\left\{q_1\times1,\ q_2\times6,\ q_3\times1\right\}.
 \label{eq:42-middle-spectrum}
\end{equation}
A shell-by-shell Hamming-symmetric permutation would require $A_{01}$ and
$A_{10}$ to receive the same number of copies of every shell eigenvalue.
Equivalently, every multiplicity in Eq.~\eqref{eq:42-middle-spectrum} would
have to be divisible by two. The multiplicities of $q_1$ and $q_3$ are both
one, so this is impossible.

For example, a capacity-respecting assignment that distributes the six
copies of $q_2$ equally is
\begin{align}
 A_{01}&\leftarrow\left\{q_1,\ q_2\times3\right\},\\
 A_{10}&\leftarrow\left\{q_2\times3,\ q_3\right\}.
 \label{eq:42-asymmetric-blocks}
\end{align}
Its trace difference is
\begin{equation}
 \Tr A_{01}-\Tr A_{10}=q_1-q_3,
 \label{eq:42-trace-difference}
\end{equation}
which is nonzero for $0<P_1<1/2$. Hence this optimal permutation gives
different excited-state populations, and therefore different local
temperatures, to the two targets. Applying the target-only Hadamard $F_2$ on
$\operatorname{span}\{\ket{01},\ket{10}\}$ replaces the two diagonal blocks
by
\begin{equation}
 A'_{01}=A'_{10}=\frac{A_{01}+A_{10}}{2},
 \label{eq:42-hadamard-average}
\end{equation}
while leaving the target-weight-zero and target-weight-two sectors
unchanged. The correction restores equal target temperatures, preserves the
passive target-energy minimum, and adds no work.

\section{Proof of the exact local-balance criterion}
\label{app:exact-local-balance-proof}
\emph{Necessity.} If a temperature-independent optimal permutation has
identical target marginals, Eq.~\eqref{eq:shellwise-local-balance}
gives an integer $d_k$ such that $d_{\alpha,k}=d_k$ for every target.
Therefore
\begin{equation}
 Md_k=\sum_{\alpha=1}^{M}d_{\alpha,k}
 =\sum_{w=0}^{M}wc_{w,k}.
 \label{eq:necessity-aggregate-degree}
\end{equation}
Lemma~\ref{lem:passive-overlap-divisibility} then yields
$M\mid wc_{w,k}$ for every $w,k$. No sector-resolved assumption is made
in this direction.

\emph{Sufficiency.} Fix $w\in\{1,\ldots,M-1\}$ and identify each
$x\in\mathcal X_w$ with the edge
\begin{equation}
 E_x=\{\alpha\in\{1,\ldots,M\}:x_\alpha=1\}.
 \label{eq:hyperedge-identification}
\end{equation}
Since each target string labels an ancillary block of dimension
$L=2^{N-M}$, the available positions form the complete $w$-uniform
multihypergraph $LK_M^{(w)}$, containing $L$ copies of every edge.
Color an edge copy by $k$ when the corresponding block position
receives $q_k$. If $\mathcal C_k$ is the color class, then
\begin{equation}
 |\mathcal C_k|=c_{w,k},\qquad
 \deg_{\mathcal C_k}(\alpha)=d^{(w)}_{\alpha,k}.
 \label{eq:degree-incidence-identity}
\end{equation}
A color class is \emph{regular} if all vertices $\alpha\in\{1,\dots,M\}$ have equal degree. It
is \emph{almost regular}, or \emph{quasi-regular}, if
\begin{equation}
 |\deg_{\mathcal C_k}(\alpha)-\deg_{\mathcal C_k}(\gamma)|\leq1
 \quad\text{for all }\alpha,\gamma.
 \label{eq:almost-regular-definition}
\end{equation}
An edge coloring is regular, respectively almost regular, when every
color class has that property. Regularity of $\mathcal C_k$ is exactly
sector-resolved one-body balance for $(w,k)$.

We use the following direct specialization of Bryant's main theorem.
For nonnegative integers $c_{w,0},\ldots,c_{w,N}$ satisfying
$\sum_kc_{w,k}=L\binom Mw$, the edge copies of $LK_M^{(w)}$ admit a
coloring such that
\begin{equation}
 |\mathcal C_k|=c_{w,k},\qquad
 |\deg_{\mathcal C_k}(\alpha)-\deg_{\mathcal C_k}(\gamma)|\leq1.
 \label{eq:almost-regular-coloring}
\end{equation}
Indeed, initially color the edge copies arbitrarily with the
prescribed class sizes and apply Bryant's Theorem~1 with
$S=V(LK_M^{(w)})$~\cite{Bryant2016}. Every permutation of $S$ is an
automorphism because every $w$-edge has multiplicity $L$. Bryant's
recoloring preserves the class sizes and makes every class almost
regular. For $L=1$, the prescribed-size statement is also Bryant's
Corollary~2.

The average degree of $\mathcal C_k$ is
\begin{equation}
 d_{w,k}=\frac{wc_{w,k}}{M}.
 \label{eq:baranyai-average-degree}
\end{equation}
Under Eq.~\eqref{eq:exact-local-balance-criterion} this is integral.
Almost regularity then forces every degree to equal the average, so
\begin{equation}
 d^{(w)}_{\alpha,k}=\deg_{\mathcal C_k}(\alpha)
 =d_{w,k}=\frac{wc_{w,k}}{M}.
 \label{eq:baranyai-exact-balance}
\end{equation}
The coloring partitions the complete edge multiset of $LK_M^{(w)}$
into the color classes $\mathcal C_k$. Hence every edge copy,
representing one computational-basis position in an ancillary block,
belongs to exactly one color class and receives exactly one shell
label. Since the edge $E_x$ occurs $L$ times, the $L$ corresponding
positions are partitioned among the shell colors, and therefore
$\sum_k n^{(w)}_{x,k}=L$, while the class sizes give
$\sum_x n^{(w)}_{x,k}=c_{w,k}$. Thus all block capacities and passive
shell multiplicities are respected. Repeating the construction in
every nontrivial sector gives the required optimal permutation.

\section{Cumulative congruence for passive shell overlaps}
\label{app:cumulative-congruence}
Recall from Sec.~\ref{sec:cooling-curve} that shell $k$ occupies
$(B_{k-1},B_k]$, sector $w$ occupies $(G_{w-1},G_w]$, and
\begin{equation}
 c_{w,k}=|(B_{k-1},B_k]\cap(G_{w-1},G_w]|.
 \label{eq:appendix-overlap-recall}
\end{equation}
Let $F(r)$ be the total target Hamming weight carried by the first $r$
passive positions. If $G_{w-1}\leq r\leq G_w$, then
\begin{equation}
 F(r)=2^{N-M}\sum_{j=0}^{w-1}j\binom Mj+w(r-G_{w-1}).
 \label{eq:cumulative-weight-function}
\end{equation}
Since $j\binom Mj=M\binom{M-1}{j-1}$,
\begin{equation}
 F(r)\equiv w(r-G_{w-1})\pmod M.
 \label{eq:cumulative-congruence}
\end{equation}

\begin{lemma}[Passive-overlap divisibility]
\label{lem:passive-overlap-divisibility}
If $M\mid\sum_wwc_{w,k}$ for every $k$, then
$M\mid wc_{w,k}$ for every $w,k$.
\end{lemma}
\begin{proof}
Summing the aggregate condition over shells $0,\ldots,k$ gives
$M\mid F(B_k)$. If $B_k\in[G_{w-1},G_w]$, then
Eq.~\eqref{eq:cumulative-congruence} gives
\begin{equation}
 M\mid w(B_k-G_{w-1}).
 \label{eq:boundary-offset-divisibility}
\end{equation}
Let
\begin{equation}
 g_w=\gcd(M,w),\qquad l_w=\frac{M}{g_w},\qquad
 s_w=\frac{w}{g_w}.
 \label{eq:gqs-definition}
\end{equation}
Then $\gcd(l_w,s_w)=1$. Writing $M=g_wl_w$ and $w=g_ws_w$,
Eq.~\eqref{eq:boundary-offset-divisibility} becomes
\begin{equation}
 l_w\mid s_w(B_k-G_{w-1}).
 \label{eq:coprime-boundary-product}
\end{equation}
Euclid's lemma therefore implies
\begin{equation}
 l_w\mid B_k-G_{w-1}.
 \label{eq:boundary-offset-q-divisibility}
\end{equation}
Thus every thermal boundary lying in sector $w$, measured from its
lower boundary $G_{w-1}$, has an offset divisible by $l_w$.

The full sector length is also divisible by $l_w$, because
\begin{equation}
 w(G_w-G_{w-1})
 =w2^{N-M}\binom Mw
 =M2^{N-M}\binom{M-1}{w-1}.
 \label{eq:sector-length-divisibility}
\end{equation}
Relative to $G_{w-1}$, every endpoint of the overlap in
Eq.~\eqref{eq:appendix-overlap-recall} is therefore divisible by
$l_w$. Their difference $c_{w,k}$ is divisible by $l_w$, which is
equivalent to $M\mid wc_{w,k}$.
\end{proof}
The converse follows by summing over $w$. Hence aggregate and
sectorwise divisibility are equivalent for passive overlaps. An
individual solution may display cross-sector cancellation, but whenever
one exists, a sector-resolved balanced representative exists as well.

\section{Constructing a sector-resolved balanced representative}
\label{app:local-balance-pseudocode}
The divisibility condition in Theorem~\ref{thm:exact-local-balance}
decides existence without numerical optimization. When an explicit
block assignment is required, the following integer formulation
constructs a sector-resolved balanced representative. Bryant's theorem guarantees the existence of a sector-resolved
assignment under Eq.~\eqref{eq:exact-local-balance-criterion}, whereas
the following integer formulation provides a convenient finite-size
representation from which explicit block multiplicities can be
extracted computationally.

For fixed $w$, let $\mathcal X_w$ be the set of target strings of
Hamming weight $w$. For every $x\in\mathcal X_w$ and
$k=0,\ldots,N$, define
\begin{equation}
 n^{(w)}_{x,k}\in\mathbb N_0
 \label{eq:sectorwise-variable}
\end{equation}
as the number of copies of $q_k$ assigned to block $A_x$. The
sector-$w$ feasibility problem is
\begin{align}
 \sum_{x\in\mathcal X_w}n^{(w)}_{x,k}&=c_{w,k},
 &&k=0,\ldots,N,
 \label{eq:sectorwise-shell-content}\\
 \sum_{k=0}^{N}n^{(w)}_{x,k}&=2^{N-M},
 &&x\in\mathcal X_w,
 \label{eq:sectorwise-block-capacity}\\
 \sum_{x\in\mathcal X_w}(x_\alpha-x_1)n^{(w)}_{x,k}&=0,
 &&\alpha=2,\ldots,M,\quad k=0,\ldots,N.
 \label{eq:sectorwise-one-body-balance}
\end{align}
The first equations reproduce the passive shell multiplicities, the
second fill every ancillary block to capacity, and the third impose
sector-resolved one-body balance by comparing every target position
with target $1$. 

\section{Explicit strict separation of the permutation criteria}
\label{app:strict-separation}
We work out the first pair, $(N,M)=(18,4)$, that satisfies the exact
local-balance criterion but fails the stronger Hamming-symmetric
criterion. The separation occurs in the target-weight-two sector,
whose six target strings correspond to the six edges of $K_4$.

Each target string labels an ancillary block of dimension
\begin{equation}
 L=2^{18-4}=2^{14}=16384.
 \label{eq:184-block-dimension}
\end{equation}
The target sectors of weights zero and one occupy
\begin{equation}
 \left [\binom40+\binom41 \right ]2^{14}=81920
 \label{eq:184-capacity-before-w2}
\end{equation}
positions in the passively ordered spectrum. The complete thermal shells $k=0,\ldots,7$ contain
\begin{equation}
 \sum_{k=0}^{7}\binom{18}{k}=63004
 \label{eq:184-shells-zero-seven}
\end{equation}
eigenvalues. Hence the lower sectors also consume $81920-63004=18916$ of the $\binom{18}{8}=43758$ copies of $q_8$, leaving $24842$ copies when the weight-two sector begins.

The weight-two sector has capacity $\binom42L=98304$. After the remaining $24842$ copies of $q_8$ and all $\binom{18}{9}=48620$ copies of $q_9$ are assigned, $24842$ positions remain for $q_{10}$. Thus
\begin{equation}
 c_{2,8}=24842,\qquad c_{2,9}=48620,\qquad c_{2,10}=24842.
 \label{eq:184-sector-spectrum}
\end{equation}
Hamming symmetry would require divisibility by $\binom42=6$, but
\begin{equation}
 24842\equiv48620\equiv2\pmod6.
 \label{eq:184-hamming-failure}
\end{equation}
The exact criterion nevertheless holds because $4\mid2c_{2,k}$. The required sector-resolved degrees are
\begin{align}
 d_{2,8}=d_{2,10}&=12421,\qquad
 d_{2,9}=24310.
 \label{eq:184-required-degrees}
\end{align}

For completeness, the other nontrivial target-weight sectors do not
contribute to the strict separation. In the $w=1$ sector, the passive
shell multiplicities are
\begin{equation}
 c_{1,6}=14796,\qquad c_{1,7}=31824,\qquad c_{1,8}=18916,
 \label{eq:184-weight-one-spectrum}
\end{equation}
whereas in the $w=3$ sector they are
\begin{equation}
 c_{3,10}=18916,\qquad c_{3,11}=31824,\qquad c_{3,12}=14796.
 \label{eq:184-weight-three-spectrum}
\end{equation}
All these multiplicities are divisible by four. Moreover,
$\binom41=\binom43=4$, and both $1$ and $3$ are coprime with
$4$. Therefore, in the $w=1$ and $w=3$ sectors, respectively, the
exact local-balance conditions $4\mid c_{1,k}$ and $4\mid3c_{3,k}$ are
equivalent to the Hamming-symmetric conditions $4\mid c_{1,k}$ and
$4\mid c_{3,k}$. Identical shell content can consequently be assigned
to the four blocks of either sector, providing Hamming symmetry and,
in particular, sector-resolved local balance. The strict separation
between the two criteria is therefore entirely due to the
target-weight-two sector.

The relation to the edge-coloring proof is especially transparent through the three perfect matchings of $K_4$,
\begin{equation}
 F_1=\{1-2,3-4\},\qquad F_2=\{1-3,2-4\},\qquad F_3=\{1-4,2-3\}.
 \label{eq:184-perfect-matchings}
\end{equation}
The matchings are displayed in
Fig.~\ref{fig:184-perfect-matchings}. Each $F_i$ contains two disjoint
edges and covers every target vertex exactly once; hence any nonnegative integer combination $a_1F_1+a_2F_2+a_3F_3$ is regular of degree $a_1+a_2+a_3$.
  \begin{figure}[t] \centering
    \begin{minipage}{0.29\columnwidth}
    \centering \begin{tikzpicture}[ baseline=(current bounding
        box.center), scale=1.26, vertex/.style={ circle, draw,
          fill=white, minimum size=5.2mm, inner sep=0pt,
          font=\scriptsize }, match/.style={ line width=1.35pt,
          MidnightBlue } ] \node[vertex] (a1) at (0,1.35) {$1$};
      \node[vertex] (a2) at (1.35,1.35) {$2$}; \node[vertex] (a3) at
      (1.35,0) {$3$}; \node[vertex] (a4) at (0,0) {$4$}; \draw[match]
      (a1)--(a2); \draw[match]
      (a3)--(a4); \end{tikzpicture} \smallskip
    \(F_1=\{1-2,3-4\}\) \end{minipage} \hfill
  \begin{minipage}{0.29\columnwidth}
    \centering \begin{tikzpicture}[ baseline=(current bounding
        box.center), scale=1.26, vertex/.style={ circle, draw,
          fill=white, minimum size=5.2mm, inner sep=0pt,
          font=\scriptsize }, match/.style={ line width=1.35pt,
          MidnightBlue } ] \node[vertex] (b1) at (0,1.35) {$1$};
      \node[vertex] (b2) at (1.35,1.35) {$2$}; \node[vertex] (b3) at
      (1.35,0) {$3$}; \node[vertex] (b4) at (0,0) {$4$}; \draw[match]
      (b1)--(b3); \draw[match]
      (b2)--(b4); \end{tikzpicture} \smallskip
    \(F_2=\{1-3,2-4\}\) \end{minipage} \hfill
  \begin{minipage}{0.29\columnwidth} \centering \begin{tikzpicture}[
        baseline=(current bounding box.center), scale=1.26,
        vertex/.style={ circle, draw, fill=white, minimum size=5.2mm,
          inner sep=0pt, font=\scriptsize }, match/.style={ line
          width=1.35pt, MidnightBlue } ] \node[vertex] (c1) at
      (0,1.35) {$1$}; \node[vertex] (c2) at (1.35,1.35) {$2$};
      \node[vertex] (c3) at (1.35,0) {$3$}; \node[vertex] (c4) at
      (0,0) {$4$}; \draw[match] (c1)--(c4); \draw[match]
      (c2)--(c3); \end{tikzpicture} \smallskip \(F_3=\{1-4,2-3\}\)
  \end{minipage}
\caption{The three perfect matchings of $K_4$ used in the explicit
  $(N,M)=(18,4)$ construction. Target positions are vertices, and
  every edge represents a weight-two target string. Each matching
  $F_i$ contains two disjoint edges and covers every target position
  exactly once. Hence each $F_i$ is one-regular, and any multiset
  combination $a_1F_1+a_2F_2+a_3F_3$ is regular of degree
  $a_1+a_2+a_3$.}
\label{fig:184-perfect-matchings}
\end{figure}

The color classes associated with the thermal shells can be chosen as
\begin{align}
 \mathcal C_8&=12421F_2,\\
 \mathcal C_9&=16384F_1+3963F_2+3963F_3,\\
 \mathcal C_{10}&=12421F_3.
 \label{eq:184-color-classes}
\end{align}
Their vertex degrees are $12421$, $24310$, and $12421$, respectively, independently of the vertex.

Using
\begin{equation}
 1100\leftrightarrow1-2,\ 1010\leftrightarrow1-3,\ 1001\leftrightarrow1-4,\ 0110\leftrightarrow2-3,\ 0101\leftrightarrow2-4,\ 0011\leftrightarrow3-4,
 \label{eq:184-string-edge-map}
\end{equation}
the color classes correspond to
\begin{align}
 A_{0011},A_{1100}&\leftarrow\{q_9\times16384\},\\
 A_{0101},A_{1010}&\leftarrow\{q_8\times12421,q_9\times3963\},\\
 A_{0110},A_{1001}&\leftarrow\{q_9\times3963,q_{10}\times12421\}.
 \label{eq:184-explicit-assignment}
\end{align}
Every block contains $16384$ entries. The $q_8$ and $q_{10}$ totals are each $2\times12421=24842$, while $2\times16384+4\times3963=48620$ gives the required $q_9$ total.

The multiplicity of $q_k$ displayed in block $A_x$ is precisely
$n^{(2)}_{x,k}$. Equations~\eqref{eq:184-color-classes} and
\eqref{eq:184-explicit-assignment} therefore define a feasible
certificate for the constructive formulation in
Appendix~\ref{app:local-balance-pseudocode}: they reproduce the
prescribed shell multiplicities, fill every block to capacity, and
satisfy
\begin{equation}
 d^{(2)}_{\alpha,k}=d_{2,k}
 \qquad\text{for every }\alpha
 \text{ and }k\in\{8,9,10\}.
 \label{eq:184-sector-balance}
\end{equation}

The construction is not Hamming-symmetric because the six block traces fall into three generally distinct pairs. Nevertheless, Eq.~\eqref{eq:184-sector-balance} shows that all four targets receive identical shell-wise incidences. Hence
\begin{equation}
 \{\text{Hamming-symmetric optimal permutations}\}\subsetneq\{\text{locally balanced optimal permutations}\}.
 \label{eq:184-strict-inclusion}
\end{equation}

\section{Chernoff--Hoeffding estimate for $\xi_M^\star$}
\label{app:chernoff-hoeffding}
We derive here the lower bound in Eq.~\eqref{eq:hoeffding-xi}. Let
\begin{equation}
 X\sim\operatorname{Bin}\!\left(N,\frac12\right).
 \label{eq:binomial-random-variable}
\end{equation}
The defining condition in Eq.~\eqref{eq:xi-definition} can be written as
\begin{equation}
 \Pr(X<\xi)
 =2^{-N}\sum_{k=0}^{\xi-1}\binom Nk
 \leq 2^{-M}.
 \label{eq:xi-binomial-condition}
\end{equation}
Thus, any integer $\xi$ satisfying
$\Pr(X<\xi)\leq2^{-M}$ is admissible in the maximization defining
$\xi_M^\star$.

For a binomial random variable with mean $N/2$, the Hoeffding lower-tail
bound is \cite{BoucheronLugosiMassart2013}
\begin{equation}
 \Pr\!\left(X-\frac N2\leq-t\right)
 \leq\exp\!\left(-\frac{2t^2}{N}\right),
 \qquad t\geq0.
 \label{eq:hoeffding-lower-tail}
\end{equation}
Choose
\begin{equation}
 t_M=\sqrt{\frac{NM\ln2}{2}}.
 \label{eq:hoeffding-choice-t}
\end{equation}
Then
\begin{equation}
 \exp\!\left(-\frac{2t_M^2}{N}\right)
 =e^{-M\ln2}=2^{-M}.
 \label{eq:hoeffding-two-minus-M}
\end{equation}
Now define
\begin{equation}
 \xi_0=\left\lfloor\frac N2-t_M\right\rfloor.
 \label{eq:hoeffding-xi-zero}
\end{equation}
If $\xi_0\geq1$, the event $\{X<\xi_0\}$ is contained in
$\{X\leq N/2-t_M\}$. Therefore,
\begin{align}
 \Pr(X<\xi_0)
 &\leq\Pr\!\left(X-\frac N2\leq-t_M\right)\nonumber\\
 &\leq e^{-2t_M^2/N}
 =2^{-M}.
 \label{eq:hoeffding-admissibility}
\end{align}
Consequently, $\xi_0$ is admissible in
Eq.~\eqref{eq:xi-binomial-condition}, and maximality of
$\xi_M^\star$ gives
\begin{equation}
 \xi_M^\star(N)
 \geq\xi_0
 =\left\lfloor\frac N2-
 \sqrt{\frac{NM\ln2}{2}}\right\rfloor,
 \label{eq:hoeffding-bound-derived}
\end{equation}
which is Eq.~\eqref{eq:hoeffding-xi}. If $\xi_0\leq0$, the displayed
lower bound is vacuous, whereas $\xi_M^\star\geq1$ by definition.

For fixed $M$ and $N\to\infty$, Eq.~\eqref{eq:hoeffding-bound-derived}
shows that
\begin{equation}
 \xi_M^\star(N)=\frac N2-O(\sqrt N),
 \label{eq:xi-fixed-M-scaling}
\end{equation}
because the opposite estimate $\xi_M^\star(N)\leq N/2+1$ follows from the
symmetry of the binomial distribution. Together with
$\xi_1^\star(\lfloor N/M\rfloor)=N/(2M)+O(1)$, this proves the asymptotic
ratio in Eq.~\eqref{eq:parallel-asymptotic-ratio}.

\section{Worked finite-size examples}
\label{app:worked-examples}

\subsection{$(N,M)=(5,2)$: a Hamming-symmetric permutation}
Each ancillary block has dimension $2^{5-2}=8$. The passive assignment is
\begin{align}
 A_{00}&\leftarrow\{q_0\times1,q_1\times5,q_2\times2\},\\
 A_{01}=A_{10}&\leftarrow\{q_2\times4,q_3\times4\},\\
 A_{11}&\leftarrow\{q_3\times2,q_4\times5,q_5\times1\}.
 \label{eq:52-block-assignment}
\end{align}
Thus the two weight-one blocks have identical shell content and no coherent correction is required. This balanced assignment is displayed in Fig.~\ref{N5M2}. The common population is
\begin{equation}
 p_\star(5,2;P_1)=4q_2+6q_3+5q_4+q_5.
 \label{eq:52-population}
\end{equation}
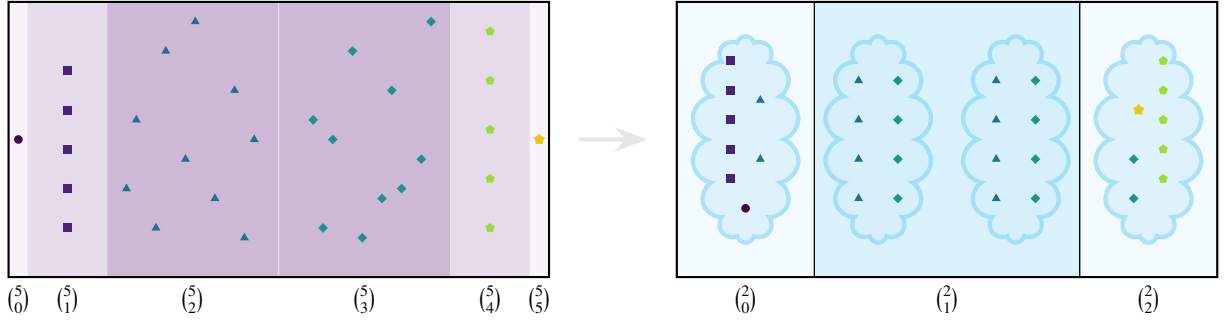
\begin{figure*}[]
    \centering
    \begin{tikzpicture}[scale=1.3, font=\sffamily, transform shape]
            
            \def\boxH{2.8} \def\boxW{5.5} \def\sep{6.8}

            \definecolor{p0}{RGB}{68, 1, 84} \definecolor{p1}{RGB}{72, 36, 117}
            \definecolor{p2}{RGB}{38, 110, 142} \definecolor{p3}{RGB}{33, 145, 140}
            \definecolor{p4}{RGB}{159, 218, 58} \definecolor{p5}{RGB}{230, 200, 20}
            \definecolor{violaBase}{RGB}{90, 20, 120} 

            \tikzset{
                m0/.style={circle, fill=p0, minimum size=2.5pt, inner sep=0pt},
                m1/.style={rectangle, fill=p1, minimum size=2.5pt, inner sep=0pt},
                m2/.style={regular polygon, regular polygon sides=3, fill=p2, minimum size=3pt, inner sep=0pt},
                m3/.style={diamond, fill=p3, minimum size=3pt, inner sep=0pt},
                m4/.style={regular polygon, regular polygon sides=5, fill=p4, minimum size=2.8pt, inner sep=0pt},
                m5/.style={star, star points=5, fill=p5, minimum size=3.5pt, inner sep=0pt}
            }

            \draw[fill=violaBase, fill opacity=0.05, draw=black!10] (0,0) rectangle (0.2,\boxH);
            \draw[fill=violaBase, fill opacity=0.15, draw=black!10] (0.2,0) rectangle (1.0,\boxH);
            \draw[fill=violaBase, fill opacity=0.30, draw=black!10] (1.0,0) rectangle (2.75,\boxH);
            \draw[fill=violaBase, fill opacity=0.30, draw=black!10] (2.75,0) rectangle (4.5,\boxH);
            \draw[fill=violaBase, fill opacity=0.15, draw=black!10] (4.5,0) rectangle (5.3,\boxH);
            \draw[fill=violaBase, fill opacity=0.05, draw=black!10] (5.3,0) rectangle (5.5,\boxH);

            \node[m0] at (0.1, 1.4) {};
            \foreach \y in {0.5,0.9,1.3,1.7,2.1} \node[m1] at (0.6, \y) {};
            \foreach \x/\y in {1.5/0.5, 2.1/0.8, 1.8/1.2, 1.3/1.6, 2.3/1.9, 1.6/2.3, 1.9/2.6, 2.5/1.4, 1.2/0.9, 2.4/0.4} \node[m2] at (\x, \y) {};
            \foreach \x/\y in {3.2/0.5, 3.8/0.8, 4.2/1.2, 3.1/1.6, 3.9/1.9, 3.5/2.3, 4.3/2.6, 3.3/1.4, 4.0/0.9, 3.6/0.4} \node[m3] at (\x, \y) {};
            \foreach \y in {0.5,1.0,1.5,2.0,2.5} \node[m4] at (4.9, \y) {};
            \node[m5] at (5.4, 1.4) {};
            \draw[thick] (0,0) rectangle (\boxW, \boxH);

            \node[below, scale=0.7] at (0.1, 0) {$\binom{5}{0}$};
            \node[below, scale=0.7] at (0.6, 0) {$\binom{5}{1}$};
            \node[below, scale=0.7] at (1.87, 0) {$\binom{5}{2}$};
            \node[below, scale=0.7] at (3.62, 0) {$\binom{5}{3}$};
            \node[below, scale=0.7] at (4.9, 0) {$\binom{5}{4}$};
            \node[below, scale=0.7] at (5.4, 0) {$\binom{5}{5}$};

            \draw[-{Stealth[scale=1.2]}, ultra thick, gray!20] (\boxW+0.3, \boxH/2) -- (\sep-0.3, \boxH/2);

            \begin{scope}[xshift=\sep cm]
                \draw[fill=cyan, fill opacity=0.05] (0,0) rectangle (1.4,\boxH);
                \draw[fill=cyan, fill opacity=0.15] (1.4,0) rectangle (4.1,\boxH);
                \draw[fill=cyan, fill opacity=0.05] (4.1,0) rectangle (5.5,\boxH);

                \node[cloud, cloud puffs=14, fill=cyan!15, draw=cyan!40, ultra thick, minimum width=1.1cm, minimum height=2.1cm, opacity=0.7] at (0.7, \boxH/2) {};
                \node[m0] at (0.7, 0.7) {}; 
                \foreach \y in {1.0,1.3,1.6,1.9,2.2} \node[m1] at (0.55, \y) {}; 
                \node[m2] at (0.85, 1.2) {}; \node[m2] at (0.85, 1.8) {};

                \node[cloud, cloud puffs=14, fill=cyan!15, draw=cyan!40, ultra thick, minimum width=1.1cm, minimum height=2.1cm, opacity=0.7] at (2.05, \boxH/2) {};
                \foreach \y in {0.8,1.2,1.6,2.0} { \node[m2] at (1.85, \y) {}; \node[m3] at (2.25, \y) {}; }

                \node[cloud, cloud puffs=14, fill=cyan!15, draw=cyan!40, ultra thick, minimum width=1.1cm, minimum height=2.1cm, opacity=0.7] at (3.45, \boxH/2) {};
                \foreach \y in {0.8,1.2,1.6,2.0} { \node[m2] at (3.25, \y) {}; \node[m3] at (3.65, \y) {}; }

                \node[cloud, cloud puffs=14, fill=cyan!15, draw=cyan!40, ultra thick, minimum width=1.1cm, minimum height=2.1cm, opacity=0.7] at (4.8, \boxH/2) {};
                \node[m3] at (4.65, 0.8) {}; \node[m3] at (4.65, 1.2) {};
                \foreach \y in {1.0,1.3,1.6,1.9,2.2} \node[m4] at (4.95, \y) {}; 
                \node[m5] at (4.7, 1.7) {};

                \node[below, scale=0.7] at (0.7, 0) {$\binom{2}{0}$};
                \node[below, scale=0.7] at (2.75, 0) {$\binom{2}{1}$};
                \node[below, scale=0.7] at (4.8, 0) {$\binom{2}{2}$};
                \draw[thick] (0,0) rectangle (\boxW, \boxH);
            \end{scope}
        \end{tikzpicture}
    \caption{Passive shell assignment for $(N,M)=(5,2)$. Left: the thermal shells $q_k=P_0^{5-k}P_1^k$ have multiplicities $\binom{5}{k}$. Right: the ordered eigenvalues fill the target sectors $w=0,1,2$. In the weight-one sector, the blocks $A_{01}$ and $A_{10}$ can receive identical shell content, so a Hamming-symmetric optimal permutation exists.}
    \label{N5M2}
\end{figure*}

\subsection{$(N,M)=(9,3)$: global versus parallel cooling}
Here every ancillary block has dimension $2^6=64$. The passive sector spectra are
\begin{align}
 w=0:&\ \{q_0\times1,q_1\times9,q_2\times36,q_3\times18\},\\
 w=1:&\ \{q_3\times66,q_4\times126\},\\
 w=2:&\ \{q_5\times126,q_6\times66\},\\
 w=3:&\ \{q_6\times18,q_7\times36,q_8\times9,q_9\times1\}.
 \label{eq:93-sector-spectra}
\end{align}
The multiplicities in sectors $w=1,2$ are divisible by three, so a Hamming-symmetric optimal permutation exists. Using Eq.~\eqref{eq:exact-cooling-curve},
\begin{equation}
 p_\star(9,3;P_1)=22q_3+42q_4+84q_5+62q_6+36q_7+9q_8+q_9.
 \label{eq:93-population}
\end{equation}
Figures~\ref{N9M3} and~\ref{N3M1parallel} compare the global and parallel resource allocations. The equal parallel benchmark uses three groups of three qubits and gives $p_\star(3,1;P_1)=3P_1^2-2P_1^3$. The two resulting populations are plotted in Fig.~\ref{N9M3comparison}.
\begin{figure}[h]
    \centering
    \begin{minipage}{0.48\textwidth}
        \centering
        \begin{tikzpicture}[scale=0.9, every node/.style={scale=0.9}]
             
    \definecolor{lightblue}{RGB}{173, 216, 230}
    \definecolor{lightred}{RGB}{255, 150, 150}
    
    \def\spacing{0.6}
    \def\circleSize{0.4cm}
    \def\Uwidth{2}
    \def\distIn{1.8}
    \def\distOut{1.8}

    \foreach \i in {0,...,8} {
        \node[circle, draw, fill=violet!80, minimum size=\circleSize, inner sep=0pt, outer sep=0pt] (qin\i) at (0, -\i*\spacing) {};
        \draw[-] (qin\i) -- (\distIn, -\i*\spacing);
    }
    
    \node[left=0.3cm] at (0, -4*\spacing) {$T$};

    \draw[fill=white!20] (\distIn, 0.3) rectangle (\distIn+\Uwidth, -8*\spacing - 0.3);
    \node at (\distIn+\Uwidth/2, -4*\spacing) {\Large $U$};

    \foreach \i in {0,...,2} {
        \node[circle, draw, fill=lightblue, minimum size=\circleSize, inner sep=0pt] (qout\i) at (\distIn+\Uwidth+\distOut, -\i*\spacing) {};
        \draw[-] (\distIn+\Uwidth, -\i*\spacing) -- (qout\i);
    }
    
    \foreach \i in {3,...,8} {
        \node[circle, draw, fill=red, minimum size=\circleSize, inner sep=0pt] (qout\i) at (\distIn+\Uwidth+\distOut, -\i*\spacing) {};
        \draw[-] (\distIn+\Uwidth, -\i*\spacing) -- (qout\i);
    }

    \node[right=0.1cm] at (\distIn+\Uwidth+\distOut, 0.5) {$T'$};
    
    \draw[decorate, decoration={brace, amplitude=5pt}, thick]
        (\distIn+\Uwidth+\distOut + 0.3, 0.2) -- (\distIn+\Uwidth+\distOut + 0.3, -2*\spacing - 0.2)
        node[midway, xshift=0.9cm] {$M=3$};
    
    \draw[decorate, decoration={brace, amplitude=5pt}, thick]
        (\distIn+\Uwidth+\distOut + 0.3, -3*\spacing + 0.2) -- (\distIn+\Uwidth+\distOut + 0.3, -8*\spacing - 0.2)
        node[midway, xshift=1.3cm] {$N-M=6$};
        \end{tikzpicture}
        \caption{Global resource allocation for $(N,M)=(9,3)$: one joint protocol uses all nine qubits to cool three targets.}
        \label{N9M3}
    \end{minipage}
    \hfill
    \begin{minipage}{0.48\textwidth}
        \centering
        \begin{tikzpicture}[scale=0.9, every node/.style={scale=0.9}]
            
    \definecolor{lightblue}{RGB}{173, 216, 230}
    \definecolor{lightred}{RGB}{255, 150, 150}
    
    \def\spacing{0.6}
    \def\circleSize{0.4cm}
    \def\Uwidth{1.6}
    \def\distIn{1.5}
    \def\distOut{1.5}

    \foreach \k in {0,1,2} {
        \pgfmathsetmacro{\yTop}{-3*\k*\spacing}
        \pgfmathsetmacro{\yMid}{-(3*\k+1)*\spacing}
        \pgfmathsetmacro{\yBot}{-(3*\k+2)*\spacing}
        \pgfmathsetmacro{\yCtr}{-(3*\k+1)*\spacing}
        \pgfmathsetmacro{\blockTop}{\yTop + 0.25}
        \pgfmathsetmacro{\blockBot}{\yBot - 0.25}

        \node[circle, draw, fill=violet!80, minimum size=\circleSize, inner sep=0pt, outer sep=0pt] (inTop\k) at (0, \yTop) {};
        \node[circle, draw, fill=violet!80, minimum size=\circleSize, inner sep=0pt, outer sep=0pt] (inMid\k) at (0, \yMid) {};
        \node[circle, draw, fill=violet!80, minimum size=\circleSize, inner sep=0pt, outer sep=0pt] (inBot\k) at (0, \yBot) {};

        \draw[-] (inTop\k.east) -- (\distIn, \yTop);
        \draw[-] (inMid\k.east) -- (\distIn, \yMid);
        \draw[-] (inBot\k.east) -- (\distIn, \yBot);

        \draw[fill=white!20] (\distIn, \blockTop) rectangle (\distIn+\Uwidth, \blockBot);
        \node at (\distIn+\Uwidth/2, \yCtr) {$U$};

        \node[circle, draw, fill=lightblue, minimum size=\circleSize, inner sep=0pt, outer sep=0pt] (outTop\k) at (\distIn+\Uwidth+\distOut, \yTop) {};
        \draw[-] (\distIn+\Uwidth, \yTop) -- (outTop\k.west);

        \node[circle, draw, fill=red, minimum size=\circleSize, inner sep=0pt, outer sep=0pt] (outMid\k) at (\distIn+\Uwidth+\distOut, \yMid) {};
        \node[circle, draw, fill=red, minimum size=\circleSize, inner sep=0pt, outer sep=0pt] (outBot\k) at (\distIn+\Uwidth+\distOut, \yBot) {};
        \draw[-] (\distIn+\Uwidth, \yMid) -- (outMid\k.west);
        \draw[-] (\distIn+\Uwidth, \yBot) -- (outBot\k.west);
    }

    \pgfmathsetmacro{\outX}{\distIn+\Uwidth+\distOut}
    \node[left=0.3cm] at (0, -3*\spacing) {$T$};
    \node[right=0.1cm] at (\outX, 0.4) {$T'$};
        \end{tikzpicture}
        \caption{Parallel resource allocation for $(N,M)=(9,3)$: the register is divided into three independent three-qubit single-target protocols.}
        \label{N3M1parallel}
    \end{minipage}
\end{figure}
\begin{figure}
    \centering
    \safeincludegraphics[width=0.5\linewidth]{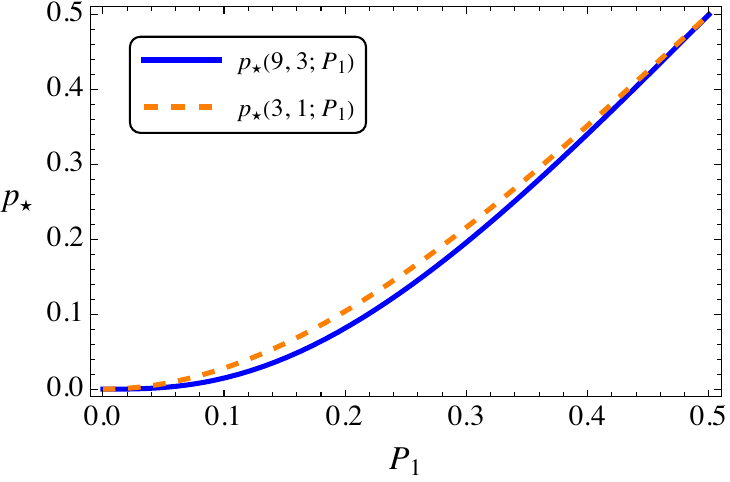}
    \caption{Common target population for $(N,M)=(9,3)$. The global optimum $p_\star(9,3;P_1)$ (blue) lies below the parallel value $p_\star(3,1;P_1)=3P_1^2-2P_1^3$ (orange).}
    \label{N9M3comparison}
\end{figure}

\subsection{$(N,M)=(6,3)$: coherent correction is required}
The sector assignment is given in Eq.~\eqref{eq:63-sector-assignment}. Since the multiplicities $14$ and $10$ are not divisible by three, a Hamming-symmetric optimal permutation does not exist. The complex-Hadamard correction nevertheless produces three identical marginals with population $p_\star(6,3;P_1)$ from Eq.~\eqref{eq:63-cooling-curve}. Figures~\ref{N6M3} and~\ref{N2M1parallel} compare the resource allocations, while Fig.~\ref{N6M3comparison} shows the corresponding populations. The parallel benchmark uses three two-qubit groups and remains at $P_1$.
\begin{figure}[h]
    \centering
    \begin{minipage}{0.48\textwidth}
        \centering
        \begin{tikzpicture}[scale=0.9, every node/.style={scale=0.9}]
            \definecolor{lightblue}{RGB}{173, 216, 230}
            \definecolor{lightred}{RGB}{255, 150, 150}
            \def\spacing{0.6}
            \def\circleSize{0.4cm}
            \def\Uwidth{2}
            \def\distIn{1.8}
            \def\distOut{1.8}
            \foreach \i in {0,...,5} {
                \node[circle, draw, fill=violet!80, minimum size=\circleSize, inner sep=0pt, outer sep=0pt] (qin\i) at (0, -\i*\spacing) {};
                \draw[-] (qin\i) -- (\distIn, -\i*\spacing);
            }
            \node[left=0.3cm] at (0, -2.5*\spacing) {$T$};
            \draw[fill=white!20] (\distIn, 0.3) rectangle (\distIn+\Uwidth, -5*\spacing - 0.3);
            \node at (\distIn+\Uwidth/2, -2.5*\spacing) {\Large $U$};
            \foreach \i in {0,...,2} {
                \node[circle, draw, fill=lightblue, minimum size=\circleSize, inner sep=0pt] (qout\i) at (\distIn+\Uwidth+\distOut, -\i*\spacing) {};
                \draw[-] (\distIn+\Uwidth, -\i*\spacing) -- (qout\i);
            }
            \foreach \i in {3,...,5} {
                \node[circle, draw, fill=red, minimum size=\circleSize, inner sep=0pt] (qout\i) at (\distIn+\Uwidth+\distOut, -\i*\spacing) {};
                \draw[-] (\distIn+\Uwidth, -\i*\spacing) -- (qout\i);
            }
            \node[right=0.1cm] at (\distIn+\Uwidth+\distOut, 0.4) {$T'$};
            \draw[decorate, decoration={brace, amplitude=5pt}, thick]
                (\distIn+\Uwidth+\distOut + 0.3, 0.2) -- (\distIn+\Uwidth+\distOut + 0.3, -2*\spacing - 0.2)
                node[midway, xshift=0.9cm] {$M=3$};
            \draw[decorate, decoration={brace, amplitude=5pt}, thick]
                (\distIn+\Uwidth+\distOut + 0.3, -3*\spacing + 0.2) -- (\distIn+\Uwidth+\distOut + 0.3, -5*\spacing - 0.2)
                node[midway, xshift=1.3cm] {$N-M=3$};
        \end{tikzpicture}
        \caption{Resource allocation for $(N,M)=(6,3)$. The global protocol has three ancillary qubits and cools all three targets, whereas the parallel benchmark consists of three two-qubit groups and cannot cool.}
        \label{N6M3}
    \end{minipage}
    \hfill
    \begin{minipage}{0.48\textwidth}
        \centering
        \begin{tikzpicture}[scale=0.9, every node/.style={scale=0.9}]
                
                \definecolor{lightblue}{RGB}{173, 216, 230}
                \definecolor{lightred}{RGB}{255, 150, 150}
                
                \def\spacing{0.65}
                \def\circleSize{0.4cm}
                \def\Uwidth{1.6}
                \def\distIn{1.5}
                \def\distOut{1.5}

                \foreach \k in {0,1,2} {
                    \pgfmathsetmacro{\yTop}{-2*\k*\spacing}
                    \pgfmathsetmacro{\yBot}{-(2*\k+1)*\spacing}
                    \pgfmathsetmacro{\yMid}{-(2*\k+0.5)*\spacing}
                    \pgfmathsetmacro{\blockTop}{\yTop + 0.25}
                    \pgfmathsetmacro{\blockBot}{\yBot - 0.25}
                    \pgfmathtruncatemacro{\klab}{\k+1}

                    \node[circle, draw, fill=violet!80, minimum size=\circleSize, inner sep=0pt, outer sep=0pt] (inTop\k) at (0, \yTop) {};
                    \node[circle, draw, fill=violet!80, minimum size=\circleSize, inner sep=0pt, outer sep=0pt] (inBot\k) at (0, \yBot) {};

                    \draw[-] (inTop\k.east) -- (\distIn, \yTop);
                    \draw[-] (inBot\k.east) -- (\distIn, \yBot);

                    \draw[fill=white!20] (\distIn, \blockTop) rectangle (\distIn+\Uwidth, \blockBot);
                    \node at (\distIn+\Uwidth/2, \yMid) {$U$};

                    \node[circle, draw, fill=violet!80, minimum size=\circleSize, inner sep=0pt, outer sep=0pt] (outTop\k) at (\distIn+\Uwidth+\distOut, \yTop) {};
                    \node[circle, draw, fill=violet!80, minimum size=\circleSize, inner sep=0pt, outer sep=0pt] (outBot\k) at (\distIn+\Uwidth+\distOut, \yBot) {};

                    \draw[-] (\distIn+\Uwidth, \yTop) -- (outTop\k.west);
                    \draw[-] (\distIn+\Uwidth, \yBot) -- (outBot\k.west);
                }

                \pgfmathsetmacro{\outX}{\distIn+\Uwidth+\distOut}
                \node[left=0.3cm] at (0, -1.5*\spacing) {$T$};
                \node[right=0.1cm] at (\outX, 0.4) {$T'=T$};

            \end{tikzpicture}
        \caption{The same $N=6$ qubits processed by $3$ independent copies of $U$, each acting on $2$ qubits and cooling $1$ of them to $T'$. Since it is impossible to cool $1$ qubit given $2$, the parallel combination of $M=1$ unitaries does not cool.}
        \label{N2M1parallel}
    \end{minipage}
\end{figure}
\begin{figure}[h]
    \centering
    \safeincludegraphics[width=0.5\linewidth]{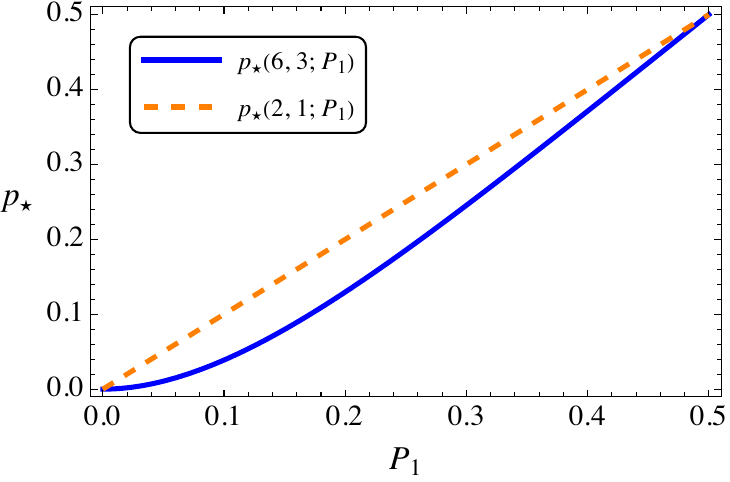}
    \caption{Common target population for $(N,M)=(6,3)$. The global value $p_\star(6,3;P_1)$ (blue) is below the parallel value $p_\star(2,1;P_1)=P_1$ (orange) for $0<P_1<1/2$.}
    \label{N6M3comparison}
\end{figure}

\subsection{$(N,M)=(6,2)$: equality with the parallel benchmark}
For two targets, each ancillary block has dimension $2^4=16$. The passive sector spectra are
\begin{align}
 w=0:&\ \{q_0\times1,q_1\times6,q_2\times9\},\\
 w=1:&\ \{q_2\times6,q_3\times20,q_4\times6\},\\
 w=2:&\ \{q_4\times9,q_5\times6,q_6\times1\}.
 \label{eq:62-sector-spectra}
\end{align}
Equation~\eqref{eq:exact-cooling-curve} gives
\begin{align}
 p_\star(6,2;P_1)
 &=3q_2+10q_3+12q_4+6q_5+q_6\\
 &=3P_1^2-2P_1^3.
 \label{eq:62-derived-equality}
\end{align}
This equals the optimal single-target population for a three-qubit group. Consequently, two parallel three-qubit protocols attain the global optimum, as stated in Theorem~\ref{thm:unique-parallel-equality}.

\end{document}